\documentclass[authoryear,preprint,12pt] 
{elsarticle}\graphicspath{ {./figures/} }

\usepackage{amssymb}
\usepackage{amsmath}
\usepackage{algorithm, algorithmic}
\usepackage{enumitem}
\usepackage{amsthm}
\newtheorem{theorem}{Theorem}
\newtheorem{lemma}{Lemma}

\theoremstyle{definition}

\usepackage{url}
\usepackage{hyperref}
\hypersetup{colorlinks=true, allcolors=blue}
\usepackage{xurl} 
\usepackage{cleveref}
\usepackage{multirow} 
\journal{Applied and Computational Harmonic Analysis}

\usepackage{todonotes,soul}
\newcommand{\M}{\mathcal M}
\newcommand{\E}{\mathbb E}
\newcommand{\R}{\mathbb{R}}
\newcommand{\OO}{\mathcal{O}}
\usepackage{tabularx}
\usepackage{booktabs}
\DeclareMathOperator*{\argmin}{arg\,min}
\DeclareMathOperator*{\argmax}{arg\,max}
\newcommand{\RNum}[1]{\uppercase\expandafter{\romannumeral #1\relax}}

\begin{document}

\begin{frontmatter}



\title{Spectral Graph Filtering for Modality-Specific Representation Learning}


\author[label1]{Shira Yoffe}
\author[label2]{Amit Moscovich}
\author[label1]{Ariel Jaffe}

\affiliation[label1]{
  organization={Department of Statistics and Data Science, Hebrew University of Jerusalem},
  country={Israel}
}
\affiliation[label2]{organization={Department of Statistics and Operations Research, Tel Aviv University},
             country={Israel}}


\begin{abstract}
Multimodal datasets, where measurements are obtained from multiple sensors, have become central to many scientific domains. 
In unsupervised settings, most representation learning methods focus on identifying shared latent structures, such as clusters or continuous processes that appear across modalities.
However, some aspects of the data may be observed only through a single modality.
For example, in computational biology, certain cell-subtypes may appear in genetic profiles but not in epigenetic markers.
In this paper, we present DELVE, a spectral method for extracting modality-specific (differential) latent variables.
Our approach constructs a graph for each modality and leverages differences in their connectivity patterns to design a graph filter that attenuates shared signals while preserving modality-specific components.
We provide an asymptotic convergence analysis for our method under a product manifold model. To evaluate the performance of our method, we test its ability to recover differential latent structures in several synthetic and real datasets.
\end{abstract}











\end{frontmatter}









\section{Introduction}\label{Sec: Introduction}
Combining sets of high-dimensional measurements from different sources has become a key task across many fields. To name but a few examples, 
recent technological developments acquire datasets of gene expression, proteomic data, and spatial information at a single-cell level \citep{moses2022museum,specht2021single}. In neuroscience, recent studies analyze PET scans of specific proteins and fMRI scans to gain insight into the development of degenerative diseases \citep{tahmasian2016based}.
In single-particle cryogenic electron microscopy, multiple views of a single protein or other macromolecule are combined to recover the shape of the molecule and its manifold of motions \citep{BendoryBartesaghiSinger2020,ToaderSigworthLederman2023,MoscovichHaleviAndenSinger2020,MajiEtal2020, ZhongEtal2021}.
In these settings, often referred to as multi-modal, the same object (i.e., an individual cell in genetics or a patient in neuroscience) is observed by two or more sensor sets, where each set may generate thousands of variables \citep{lederman2018learning, chaudhuri2009multi, xu2013survey}. 
A common goal is to provide a low-dimensional representation for downstream analysis based on the observations from the different sensors.  

For a single-modality, the manifold assumption has been the basis of numerous data analysis methods developed in recent decades \citep{coifman2006diffusion,belkin2003laplacian}. According to this assumption, the high-dimensional observations $x_1,\ldots,x_n \in \R^\ell$
are assumed to lie near the image of a smooth mapping $T:\mathbb{R}^d \to \mathbb{R}^\ell$, parameterized by unknown vectors $\theta = (\theta_1, \ldots, \theta_n)\in \mathbb{R}^d$, such that,
\begin{align}
    \theta_i \xrightarrow{T} x_i\,,
\end{align}
where often $d \ll \ell$.
A main challenge of non-linear dimensionality reduction is to compute a low-dimensional representation of $x_i$ associated with the latent vectors $\theta_i$. 
Consider the setting where two sensors, A and B, take simultaneous snapshots of some object, yielding a collection of pairs $(x^A_i, x^B_i) \in \R^{\ell_A} \times \R^{\ell_B}$ where each pair contains two measurement modalities.
Similarly to the single-modality setting, we assume that each observation can be approximated by a continuous mapping applied to a set of latent variables,
\begin{equation}\label{eq:double_manifold}
(\theta_i, \psi_i^A) \xrightarrow{T_A} x_i^A \qquad (\theta_i, \psi_i^B) \xrightarrow{T_B} x_i^B.    
\end{equation}
The latent parameters $\theta$ represent degrees of freedom that both sensors A and B can measure. In contrast, the latent variables $\psi^A$ represent degrees of freedom that only $A$ can measure and $B$ is invariant to. The opposite holds for $\psi^B$.
For example, A and B might represent two different color channels of a digital camera. Both channels are sensitive to changes in brightness, but each one is blind to changes in color outside of its respective spectrum.
Figure \ref{fig:Toy example} illustrates our setting with a toy example from \cite{lederman2018learning}.
The left image depicts an observation from modality $A$, where the camera captures two rotating figurines: a Yoda and a bulldog. The right image presents an observation from modality $B$, in which the same bulldog figurine is shown alongside a rabbit doll. In this example, the shared latent variable $\theta$ is the rotation angle of the bulldog. The modality-specific variables $\psi^A$ and $\psi^B$ are the rotation angles of the Yoda figurine and the rabbit, respectively.
Many works focus on uncovering latent variables shared across modalities while suppressing modality-specific factors. However, modality-specific signals often contain valuable information about the observed object and can improve downstream tasks such as clustering and prediction. 
For example, in multi-omics data, a subset of cells that forms a single cluster in gene expression measurements (e.g., scRNA-seq) may split into two distinct subtypes in another modality, such as ATAC-seq. 
Motivated by this observation, our primary objective is to compute low-dimensional representations that explicitly capture latent structure unique to each modality, while disentangling it from shared structure. By isolating modality-specific components, we enable more precise characterization of complex systems and improve downstream analyses in unsupervised domains.



\begin{figure}[tb]
    \centering
    \begin{minipage}[c]{0.27\linewidth}
        \vspace{0pt}
        \centering
        \includegraphics[width=\linewidth]{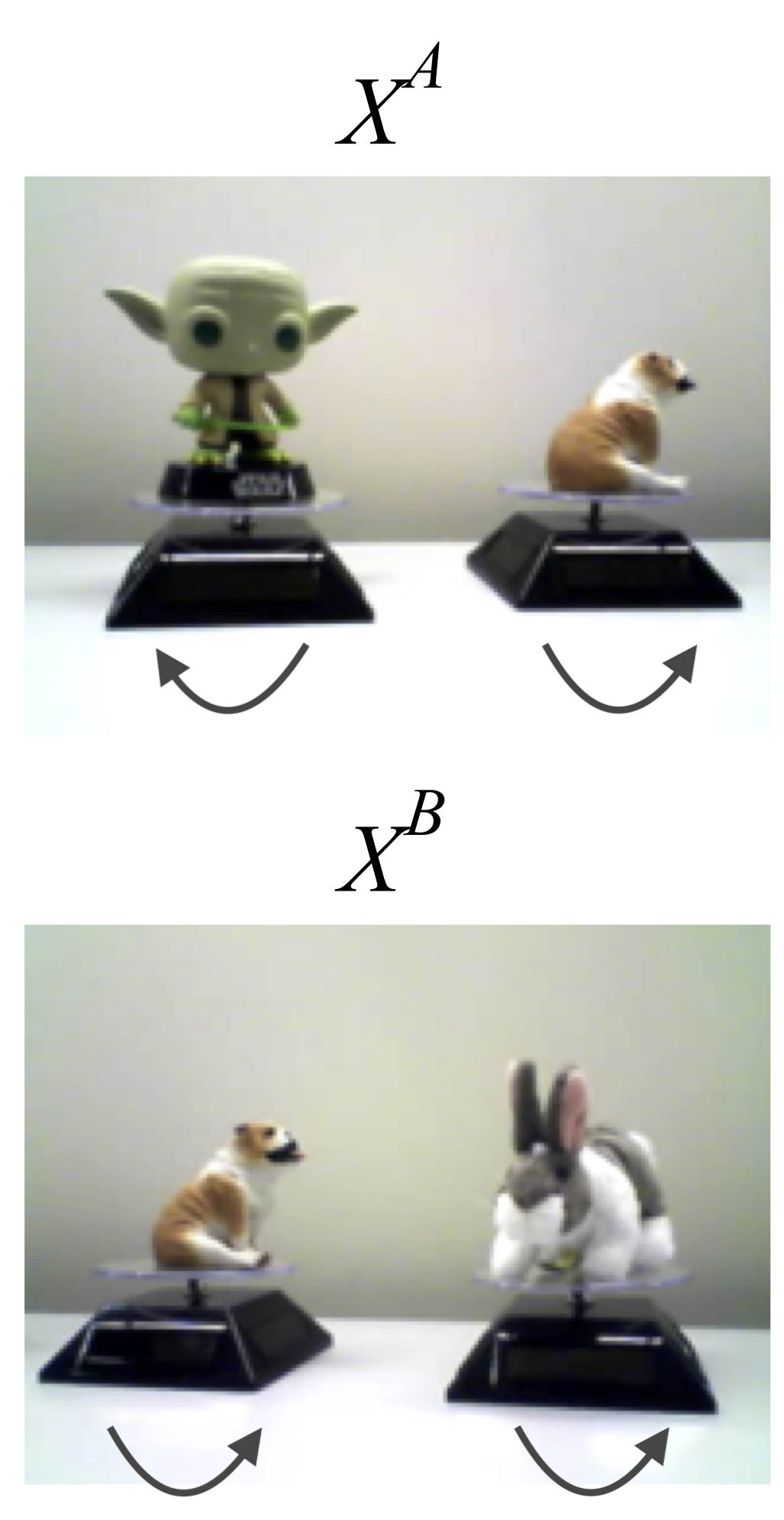}
    \end{minipage}
    \hspace{0.02\linewidth}
    \begin{minipage}[c]{0.69\linewidth}
        \vspace{0pt}
        \centering
        \includegraphics[width=0.9\linewidth]{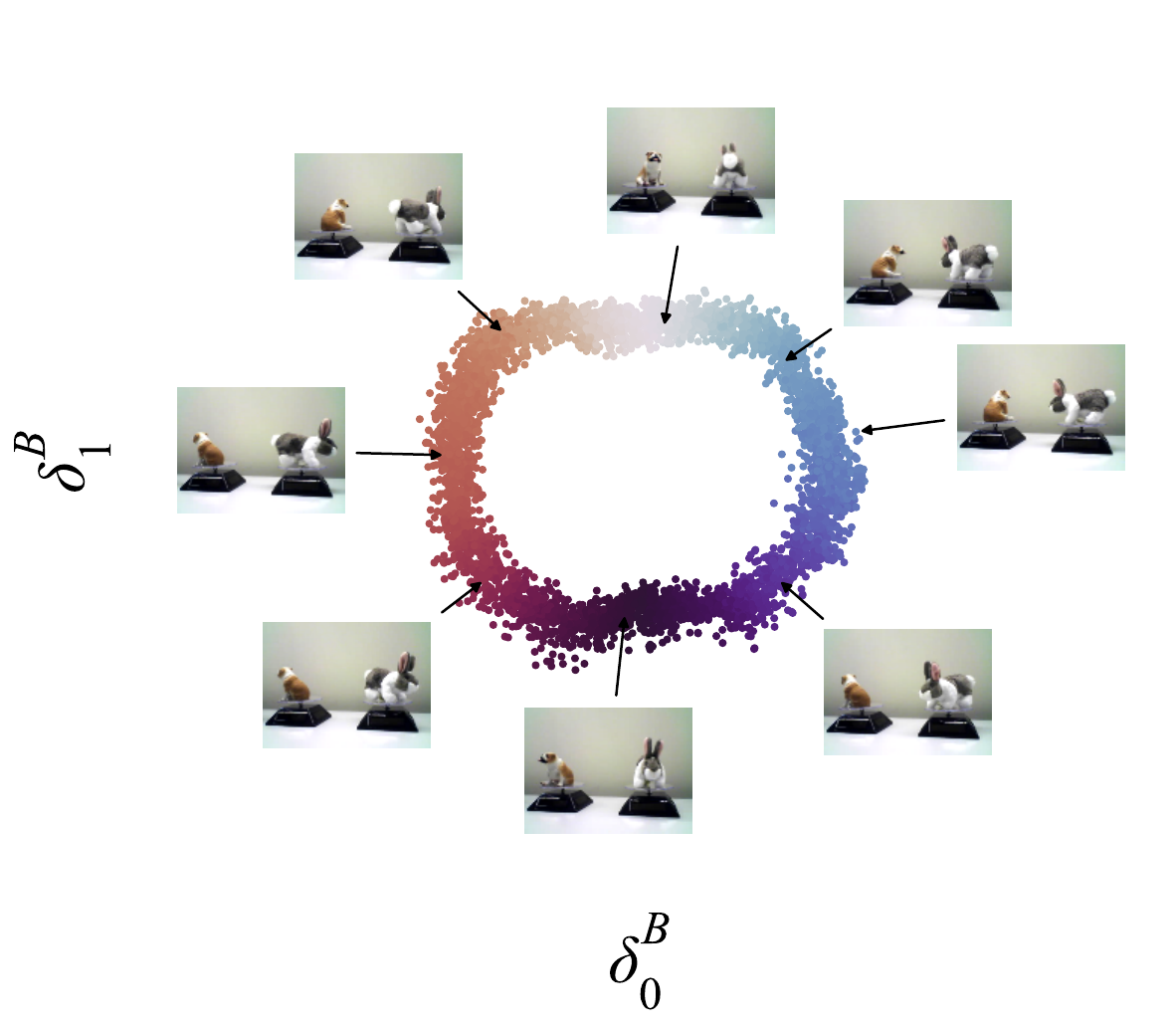}
    \end{minipage}

    \caption{Rotating dolls. \textbf{Left:} A single image from $X^A$ and $X^B$. 
    \textbf{Right:} Embedding using DELVE's two leading differential vectors, colored by the rabbit's rotation angle $\psi^B$.}
    \label{fig:Toy example}
\end{figure}



\subsection{Related work}

Many recent works have developed spectral methods to uncover latent variables shared across modalities. 
Classical approaches such as CCA and its kernel variants seek maximally correlated projections between modalities \citep{hotelling1992relations, akaho2001kernel, lai2000kernel}.
In addition, several manifold-based techniques construct a common representation by coupling geometric information from multiple modalities. Representative examples include jointly smooth functions across manifolds \citep{diethe2008multiview}, simultaneous diagonalization of Laplacians \citep{eynard2015multimodal}, and diffusion-based constructions such as local CCA kernels and locality alignment \citep{yair2016local, zhao2018multi}. Alternating Diffusion and subsequent refinements \citep{lederman2018learning, lindenbaum2020multi} also focus on recovering shared components while attenuating sensor-specific variation.

In contrast, relatively fewer spectral methods explicitly address the extraction of sensor-specific or differential variables. Classical discriminative approaches such as Linear Discriminant Analysis (LDA) \citep{fisher1936use} and its feature-selection extensions \citep{song2010feature, sharma2014feature} aim to identify separating directions in supervised settings. FKT transform \citep{fukunaga1970application} and its kernel extension \citep{li2007kernel} target differences between classes through spectral decompositions. Manifold-based methods have also been proposed for differential feature extraction: DiSC \citep{sristi2022disc} identifies feature groups that differentiate between conditions, while \citet{pmlr-v202-cohen23b} integrate Riemannian structure with spectral analysis for feature selection. Under the multi-view setting, \citet{shnitzer2019recovering} extend Alternating Diffusion \citep{lederman2018learning} to recover sensor-specific variables, and \citet{katz2020spectral} construct geodesic paths between modality-specific kernels to characterize their differences. Recent works have also designed VAE's to capture both sensor-specific and shared components in multimodal settings \citep{lee2021private}.

\paragraph{Main contributions and paper outline}
Given datasets from two modalities, we develop a simple, graph-based approach to uncover latent variables captured by one modality but invisible to the other. The method constructs a graph filter from one modality and applies it to a graph derived from the other.
We show that the eigenvectors of the filtered graph induce an embedding that isolates a single modality-specific latent variable.
We then describe an iterative procedure for recovering multiple latent variables.

On the theoretical side, we establish recovery guarantees under a product manifold model, characterizing the conditions under which our method provably identifies modality-specific latent variables.
Finally, we provide an extensive empirical evaluation on both simulated and real multimodal datasets, demonstrating that our approach reveals modality-specific structure that competing methods fail to detect.

The rest of the paper is organized as follows. In Section~\ref{Sec: Setting and background}, we present the problem setting and review the relevant background in manifold learning and graph signal processing. Section~\ref{Sec: Algorithm} introduces our algorithms for extracting a single latent variable and their extension to multiple variables. In Section~\ref{Sec: Analysis}, we provide a theoretical analysis of our method under a product manifold model. 
A computational complexity analysis is provided in Section ~\ref{sec:complexity}.
Finally, in Section~\ref{Sec: Examples}, we present results on both simulated and real multimodal datasets.

\begin{table}
\renewcommand{\arraystretch}{1.2}
\begin{tabularx}{\textwidth}{p{0.15\textwidth}p{0.62\textwidth} p{0.15\textwidth}}

\toprule
  {\bf Name} & {\bf Description} & {\bf Ref}\\\hline
  
  $n$ & Number of observations \\\hline
  
  $x_i^A$ & i-th observation in dataset A \\\hline

  ${\scriptstyle \theta}$   & Shared variable \\\hline
  ${\scriptstyle \psi^A,\psi^B}$ & Distinctive/differential variables of modalities $A$ and $B$ \\\hline
  
  ${\scriptstyle W^A, W^B}$ & Weight matrices of each graph,  & Eq. \eqref{eqn: gaussian kernel weights} \\\hline
  ${\scriptstyle P^A,P^B}$ & The symmetric operators & Eq. \eqref{Symmetric operator}\\\hline
  ${\scriptstyle L^A, L^B}$ & The symmetric Laplacian matrices & Eq. \eqref{eqn: graph Laplacian} \\\hline
  $v^A,v^B$ & Eigenvectors of ${\scriptstyle L^A, L^B}$\\\hline
  
  ${\scriptstyle \lambda^A, \lambda^B}$ & Eigenvalues of ${\scriptstyle L^A, L^B}$\\\hline

  ${\scriptstyle H(L^A), H(L^B)}$ & Graph filters & Eq.  \eqref{eqn:projection function}\\\hline

  ${\scriptstyle \Tilde{P}^A, \Tilde{P}^B}$ & The filtered operators & Eq. \eqref{eqn: filtered laplacian} \\\hline

  ${\scriptstyle \delta^A,\delta^B}$ & Differential vectors & Sections ~\ref{Sec: Algorithm - one variable} and~\ref{Sec: Algorithm - geneeralized}\\\hline


  $p$ & Uniform sampling density over the manifold ${\scriptstyle \M}$ &\\\hline
  
  ${\scriptstyle f^A, f^B} $ &  Eigenfunctions of the LB operators of ${\scriptstyle 
 \M^A,\M^B}$\\\hline

  $\eta^A,\eta^B$ &  Eigenvalues of the LB operators of ${\scriptstyle \M^A,\M^B}$\\\hline

  $\rho_X(f_k)$ &  The sampling operator, which evaluates $f_k$ at a set of points $X$\\\hline

  $\phi_k(X)$ & Normalized vector of samples & Eq. \eqref{eq: vector of samples}\\\hline

  $\| \cdot \|$ & Matrix operator norm or Euclidean vector norm & \\\hline

  $\| \cdot \|_F$ & Matrix Frobenius norm & \\
  
\bottomrule
\end{tabularx}
\caption{Notation table}
\end{table}

\section{Problem setting and related background}\label{Sec: Setting and background}

Consider two data sets $X^A \in \mathbb{R}^{n \times \ell_A}$ and $X^B \in \mathbb{R}^{n \times \ell_B}$, where the rows of both matrices consist of paired observations from two different sensors or modalities that satisfy the double manifold assumption of \eqref{eq:double_manifold}. The columns of $X^A$ and $X^B$ contain the features measured by the sensors A and B, respectively.
Our aim is to compute a low-dimensional representation that is only a function of the distinctive (or differential) latent parameters $\psi^A$ and $\psi^B$ and not the shared parameter $\theta$.
\subsection{Graph construction and spectral analysis}
Our approach is based on comparing two graphs, each constructed separately from one modality. 
\begin{align}
    G^A = (V, E^A, W^A), \qquad G^B = (V, E^B, W^B).
\end{align} 
The vertices $V = \{x_1,\ldots,x_n\}$ correspond to the $n$ observations and are the same in both graphs. The
weights $W^A, W^B \in \mathbb{R}^{n \times n}$ are computed separately by two kernel functions,
\begin{align}
    K^A: \mathbb{R}^{\ell_A} \times \mathbb{R}^{\ell_A} \rightarrow \mathbb{R},\qquad K^B: \mathbb{R}^{\ell_B} \times \mathbb{R}^{\ell_B} \rightarrow \mathbb{R}. 
\end{align}
A standard choice for the kernels $K^A, K^B$ is a symmetric Gaussian. 
\begin{align}
    \label{eqn: gaussian kernel weights}
    W^A_{i,j}
    =
    \exp \left( - \|x^A_i-x^A_j\|^2 / 2\sigma_A^2 \right),
    \qquad
    W^B_{i,j}
    =
    \exp \left( - \|x^B_i-x^B_j\|^2 / 2\sigma_B^2 \right).
\end{align}
Common variations include replacing the fixed bandwidth $\sigma$ with a density-dependent bandwidth~\citep{zelnik2004self}, replacing the Euclidean norm with other norms or metrics such as the  Mahalanobis distance or Wasserstein metric~\cite{KileelMoscovichZeleskoSinger2021}, and weights based on k-nearest-neighbor distances~\citep{ChengWu2021}.

Let $D^A$, $D^B$ be two diagonal matrices whose elements are equal to the degrees of the nodes in $G^A$, $G^B$, respectively, such that $D_{ii} = \sum_{j=1}^n{W_{ij}}$. We denote by $P^A$, $P^B$ the operators,
\begin{equation}
\label{Symmetric operator}
P^A = (D^A)^{-0.5}W^A (D^A)^{-0.5}, \quad \quad P^B = (D^B)^{-0.5}W^B (D^B)^{-0.5},
\end{equation}
and the symmetric normalized Laplacian matrices by,
\begin{equation}
    \label{eqn: graph Laplacian}
    L^A = I - P^A, \qquad  L^B = I - P^B.   
\end{equation} 
The key idea in this paper is that the degrees of freedom corresponding to the differential latent variables $\psi^A$ and $ \psi^B$ can be extracted by analyzing the differences in the connectivity patterns of $G^A$ and $G^B$. For analyzing the two graphs, we leverage tools from graph signal processing. 

\subsection{Graph signal processing and graph filters}

 Let us first describe some notations and definitions.
The signals that we consider are real functions $s: V \to \mathbb{R}$ defined on the vertices of the graph. Since the graph has $n$ vertices, we will identify them with vectors $s \in \R^n$.
 The symmetric normalized Laplacian matrix defined in Eq. \eqref{eqn: graph Laplacian} is positive semi-definite with at least one zero eigenvalue  \citep{shuman2013emerging, von2007tutorial}. We denote its eigenvalues in non-decreasing order by 
$\lambda_0 =0 \leq  \lambda_1 \leq \cdot\cdot\cdot \leq \lambda_{n-1}$ with corresponding eigenvectors
$v_i$ that form an orthonormal basis of $\mathbb{R}^n$.

The Laplacian matrix can be viewed as a smoothness functional for graph signals.
For the symmetric normalized Laplacian, we have,
\begin{equation}
    \label{eqn: Laplacian as differential operator}
     s^T L s = \frac{1}{2} \sum_{i,j = 1}^n W_{i,j} \bigg(\frac{s_i}{\sqrt{D_{i,i}}} - \frac{s_j}{\sqrt{D_{j,j}}}\bigg)^2.
\end{equation} 
Combining Eq. \eqref{eqn: Laplacian as differential operator} with the Courant-Fischer Theorem shows that the eigenvectors with the smallest eigenvalues are orthonormal minimizers of the RHS \citep[Theorem 4.2.11]{horn2012matrix},
\begin{equation}
    \label{eqn: Laplacian eigenvalues - courant-fisher}
     v_\ell = 
     \argmin_{\substack{ ||s||_2=1 ; \\ s \bot \{v_0,...,v_{\ell-1}\}}} 
     \frac{1}{2} \sum_{i,j = 1}^n W_{i,j} \bigg(\frac{s_i}{\sqrt{D_{i,i}}} - \frac{s_j}{\sqrt{D_{j,j}}}\bigg)^2.
\end{equation}
Eq. \eqref{eqn: Laplacian eigenvalues - courant-fisher} implies that for an eigenvector with a small eigenvalue, the $\sqrt{D}$-normalized values of neighboring vertices are close.
This is analogous to classic Fourier analysis, where eigenfunctions with low eigenvalues have low spatial frequencies.
Continuing this analogy, one can define the \textit{Graph Fourier Transform} $\hat{s}$ of a vector $s\in \mathbb{R}^{n}$, and its inverse counterpart, the \textit{Graph Inverse Fourier Transform} by
\begin{equation}
    \label{eqn: graph Fourier transform}
    \hat{s}_\ell = \langle v_\ell,s\rangle,
    \qquad 
    s = \sum_{\ell=0}^{n-1} \hat{s}_\ell v_\ell.
\end{equation}
Eq. \eqref{eqn: graph Fourier transform} provides a new representation for the graph signals in the spectral domain. Thus, one can define \textit{graph spectral filtering} by including eigenvalue-based weights,

\begin{equation}
    \label{eqn: graph spectral filtering}
    H(s) = \sum_{\ell=0}^{n-1} \hat{s}_\ell h(\lambda_\ell) v_\ell.
\end{equation}
For example, the low-pass filter $h(\lambda_\ell;\tau) = \mathbf{1}_{\lambda_l \leq \tau}$ removes the high-frequency components of a signal and thus acts as a smoothing filter. In the following section, we design graph high-pass filters to discover latent variables that differentiate between two datasets.

\section{DELVE: Differential Latent Variables Extraction}\label{Sec: Algorithm}

In this section, we introduce \textit{DELVE}, a spectral algorithm for uncovering modality-specific latent variables in multimodal data.
We first describe the algorithm for computing an embedding that is a function of a single sensor-specific latent variable (Section~\ref{Sec: Algorithm - one variable}) and then extend it to multiple latent variables (Section~\ref{Sec: Algorithm - geneeralized}).

\subsection{Embedding sensor-specific latent variable}\label{Sec: Algorithm - one variable}
\begin{figure}[tb]
   \includegraphics[width=0.99\linewidth]{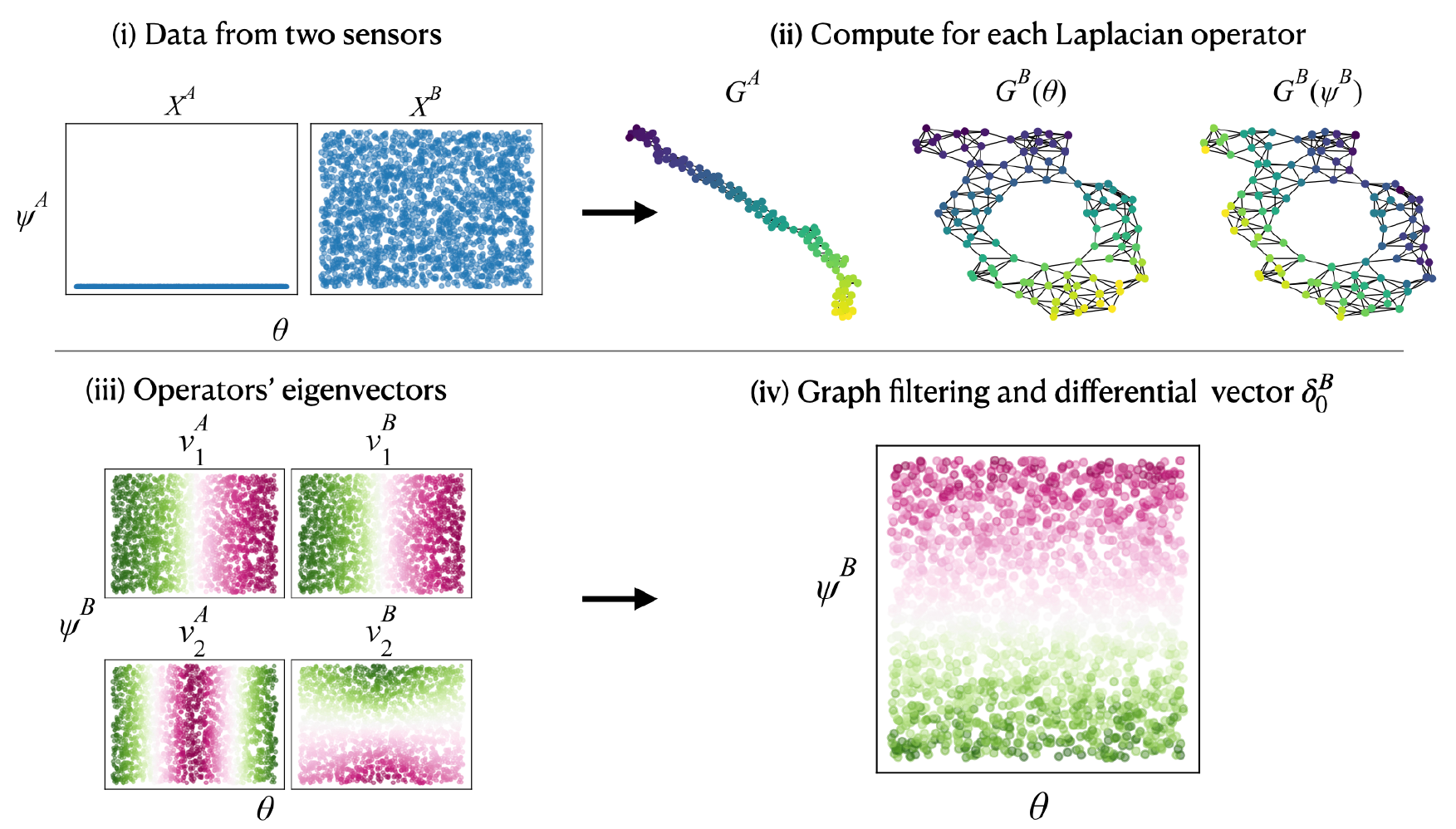}
   \caption{Illustration of Algorithm \ref{alg:description} using the line ($X^A$) vs. rectangle ($X^B$) example. (II) $G^A$ with nodes colored by $\theta$, and $G^B$ with nodes colored by $\theta$ and $\psi^B$ (III) Top: Points colored by the eigenvectors of $P^A$. Bottom: Points colored by the eigenvectors of $P^B$.
   (IV) Points colored by the leading differential vector $\delta_0^B$.}
   \label{fig:Eigenvectors of L^A and L^B}
\end{figure}

Our approach is based on uncovering differences between the spectra of the graph representations of the two modalities.
We compute the operators $P^A, P^B$ via Eq. \eqref{Symmetric operator}, the symmetric  Laplacian matrices, $L^A, L^B$ by Eq. \eqref{eqn: graph Laplacian}, and their eigenvalues and eigenvectors, denoted respectively by $(\lambda_i^A,v^A_i)$ and $(\lambda_i^B,v^B_i)$.
To compute an embedding that is a function of $\psi^B$, we design a filter   
that attenuates graph signals of the shared variable $\theta$ from $P^B$, while maintaining the signal of $\psi^B$.
The filter is based on the eigenvalues and eigenvectors of $L^A$.


Let $h(\lambda): [0,1] \to [0,1]$ denote a non-decreasing function in $\lambda$. For simplicity, consider an example where $h(\lambda)$ is equal to the step function
\begin{equation}\label{eq:threshold_filter}
h(\lambda) = \begin{cases}
    1, &  \lambda> \tau.\\
    0, & \text{otherwise}.
  \end{cases}
\end{equation}
We denote by $H(L^A)$ and $H(L^B)$ a filtering operator defined by,
\begin{equation}
    \label{eqn:projection function}
    H(L^A) =   \sum_i h(\lambda^A_i) v^A_i(v^A_i)^T, \qquad \qquad
    H(L^B) = \sum_ih(\lambda^B_i)v^B_i (v^B_i)^T.
\end{equation}
For the step function filter in Eq.~\eqref{eq:threshold_filter}, applying $H(L^A)$ to a vector $y$ attenuates any component in $y$ correlated with the leading eigenvectors of $L^A$. Our next step is to compute the filtered operators $\Tilde{P}^B,\Tilde{P}^A$ by, 
\begin{equation}
    \label{eqn: filtered laplacian}
    \Tilde{P}^B = H(L^A) P^B H(L^A),
    \qquad \qquad
    \Tilde{P}^A = H(L^B) P^A H(L^B).
\end{equation}
We denote by $\delta^A$ and $\delta^B$ the leading eigenvectors of $\Tilde{P}^A$ and $\Tilde{P}^B$, respectively. We refer to $\delta^A$ and $\delta^B$ as the  \textit{differential vectors} of $P_A$ and $P_B$, respectively.

To illustrate our approach, we use a simple example where $x^B$ contains $n$ observations sampled uniformly from a rectangle, such that $x^B_i \in [0, a] \times [0, b]$ and $x^A_i$ contains only the first coordinate of $x^B_i$, lying in $[0, a]$. 
In this toy example, the transformations $(\theta_i,\psi^B_i) \to x^B_i$ and $\theta_i \to x^A_i$ are both the identity mapping.  
Figure~\ref{fig:Eigenvectors of L^A and L^B} provides an illustration of our filtering process for this example. Panel (ii) shows the two graphs computed from $X^A$ and $X^B$, respectively. The vertices of $G^A$ are colored by  $\theta$, and the vertices of $G^B$ are colored by $\theta$ (left) and $\psi^B$ (right). 
Panel (iii) shows a scatter plot, where the two coordinates of each point $i$ are equal to $\theta_i$ (first element of $x^B_i$) and $\psi^B_i$ (second element of $x^B_i$). Each point is colored according to its value in the two leading eigenvectors of $P^A$ (upper row) and $P^B$ (bottom row). As expected, the eigenvectors of $X^A$ depend only on $\theta$. In contrast, the second eigenvector of $X^B$ depends on $\psi^B$. 
In panel (iv), the points are colored according to $\delta_0^B$, the leading eigenvector of $\Tilde{P^B}$. The vector clearly encodes the differential vector $\psi^B$.

For an additional illustration, recall the two-figurine example presented in the introduction. The right panel in Figure~\ref{fig:Toy example} shows a scatter plot, where the two coordinates of each image are equal to the corresponding elements in $\delta_0^B$ and $\delta_1^B$. The points are colored according to the rabbit's rotation angle, a modality-specific variable.
Algorithm~\ref{alg:description} summarizes the steps for extracting a single latent variable. 

\paragraph{Choice of threshold $\tau$} The choice of $\tau$ for the filtering operator $H(L^B)$ should be guided by the spectrum of $L^B$. 
Recall that the eigenvalues of $L^B$ equal one minus the eigenvalues of $P^B$. In our experiments, $\tau$ was chosen so  that the low-frequency components account for 90\% of the  sum of eigenvalues of $P^B$.
Formally,
\begin{equation}\label{eq:set_tau}
\tau = \min_{\tilde \tau} \sum_{i;\lambda_i^B \leq \tilde \tau} (1-\lambda_i^B) \geq 0.9\sum_{i=1}^n (1-\lambda_i^B)   
\end{equation}
Assuming $L^B$ exhibits a sufficient spectral decay, setting $\tau$ according to Eq. \eqref{eq:set_tau} significantly attenuates signals that correlate with $\theta$, with negligible effect on $\psi^B$.


\begin{algorithm}[tb]
    \caption{Recovering an embedding for a sensor-specific latent variable}
	\label{alg:description}
 \textbf{Input:} $X^A \in \R^{n \times \ell_A},X^B \in \R^{n\times \ell_B}$ two matrices containing $n$ samples with $\ell_A$ and $\ell_B$ 
 features respectively. 
             A filter function $h(\lambda):[0,1] \to [0,1]$. 
             \\[0.2cm]
    \textbf{Output:} Differential vectors $\delta^A, \delta^B \in \mathbb{R}^n$.\\[-0.3cm]
	\begin{algorithmic}[1] 
		\STATE Compute the weight matrices $W^A,W^B$ from $X^A, X^B$ using Eq.~\eqref{eqn: gaussian kernel weights}.
        \STATE Calculate the operators $P^A, P^B$ by Eq. \eqref{Symmetric operator} and the symmetric normalized Laplacian matrices $L^A, L^B$ by Eq. \eqref{eqn: graph Laplacian}.
        \STATE Obtain the filtering matrices $H(L^A), H(L^B)$ via Eq.~\eqref{eqn:projection function}.
        \STATE Compute the filtered operators $\tilde P^A,\tilde P^B$ according to Eq. \eqref{eqn: filtered laplacian}. 
        \STATE Compute the differential vector $\delta^A$ from $\tilde P^A$ and $\delta^B$ from $\tilde P^B$.
\end{algorithmic}
\end{algorithm}

\subsection{Embedding multiple sensor-specific latent variables}\label{Sec: Algorithm - geneeralized} 

In the previous section, we developed a graph-based method to compute a one-dimensional embedding that encodes a sensor-specific latent factor. Here, we extend this approach to compute a multidimensional embedding for each observation, where each coordinate encodes a different sensor-specific factor.
We begin with an example showing why one cannot naively compute multidimensional embeddings from the leading eigenvectors of $\Tilde P^A$ and $\Tilde P^B$. We then derive an iterative procedure that constructs embeddings whose coordinates each contribute new, non-redundant information about the observations.

Consider the case where 
$X^B$ contains $n$ observations from a 3D cube, $x^B_i \in [0, a] \times [0, b] \times [0, c] $, where $c< \min (a,b)$. Similarly to the previous example, $x^A_i$ contains only the first coordinate of $x^B_i$. Thus, $\theta$ is the first coordinate and $\psi^B_0$ and $\psi^B_1$ are the second and third coordinates. Assuming a large number of samples, $\delta_0^B$ converges to a function of $\psi^B_0$. However, the shape of the next eigenvector $\delta_1^B$ depends on the cube dimensions $(a,b)$ and $c$ and may not necessarily be a function of $\psi^B_1$. For example, $\delta_1^B$ may correspond to a mixed eigenfunction involving both $\psi^B_0$ and $\theta$, or to a higher-frequency mode of $\psi^B_0$.

Thus, an additional step is required to ensure that: (i) the next differential vector is associated only with a single differential variable, and (ii) that it is not redundant. i.e., it is associated with new variables that were not already learned.
To achieve this, we introduce an iterative method that computes, at each iteration, a new embedding for the $n$ observations that depends on $\theta$ and the previously computed sensor-specific embedding. 
Assume that we computed a one-dimensional embedding that encodes $\psi^B_1$. 
To compute an embedding that depends on $\theta$, 
we propose a similar approach to \cite{lederman2018learning}, where the embedding is computed by the leading $k_0$ eigenvectors of the symmetric product of $P^A$ and $P^B$,
\begin{equation}
    \label{eqn:projection of LA onto LB} 
    P^{\theta} = P^A P^B + P^BP^A.
\end{equation} 
We denote this embedding by $V^{(0)} \in \R^{n \times k_0}$, where $k_0$ is set according to the spectrum decay of $P^{\theta}$.
Next, we concatenate $V^{(0)}$ with $\delta^B_0$ to obtain $V^{(1)} \in \R^{n \times (k_0+1)}$. The rows of $V^{(1)}$ constitute vector embeddings that encode $\theta$ and $\psi_0^B$. 
Thus, one can think of $V^{(1)}$ and $X^B$ as a new multimodal dataset, where the shared variables are $(\theta,\psi^B_0)$ and the sensor-specific variable is $\psi^B_1$. To obtain a one-dimensional embedding for $\psi^B_1$, we apply Algorithm~\ref{alg:description} with $V^{(1)}$ and $X^B$ as inputs. Note that here we apply a one-sided version of Algorithm \ref{alg:description} that filters $X^B$ to obtain  $\delta^B$ without computing $\delta^A$. Computing additional sensor-specific variables (such as $\psi^B_2$) can be done iteratively in a similar way. The steps for applying our iterative approach are outlined in Algorithm \ref{alg:description 2}. Fig. \ref{fig: eigenvectors of L^A and L^B 3D cube} provides an illustration of this algorithm based on the cube-vs-line example.
The left panel shows the leading eigenvectors of the $\tilde P^B$ obtained by algorithm \ref{alg:description}. Importantly, $\delta_1^B$ and $\delta_2^B$ are functions of $\theta$ and the $\psi^B_0$, and do not provide new information. The right panel shows the differential vectors obtained by the iterative approach of algorithm \ref{alg:description 2}. The two vectors are functions of $\psi^B_0$ and $\psi^B_1$.

\begin{algorithm}[tb]
	\caption{Embedding multiple sensor-specific latent variables}
	\label{alg:description 2}
    \textbf{Input:} 
    Matrices $X^A \in \R^{n \times \ell_A}$ and $X^B \in \R^{n\times \ell_B}$, containing $n$ samples with $\ell_A$ and $ \ell_B$ features respectively.
    $K$ number of iterations.
    Filter function $h(\lambda):[0,1] \to [0,1]$.
    $k_0$ dimension of shared latent space.\\[0.2cm]
    \textbf{Output:} Differential vectors $\delta^{B_1},\ldots, \delta^{B_{K}} \in \R^{n}$.\\[-0.3cm]
\begin{algorithmic}[1]
            \STATE Calculate the differential vector $\delta^B_0$ via Algorithm~\ref{alg:description}
            , and denote as $\delta^{B_0}$.\\[-1em]
            \STATE Using Eq. \eqref{eqn:projection of LA onto LB}, compute $P^\theta$ and its leading eigenvectors $V^{(0)} \in \R^{n \times k_0}$.
            \FOR{$i \in\{1,\ldots,K\}$}
            \STATE Obtain $V^{(i)} \in \R^{n \times (k_0 +i)}$ by concatenating the matrix $V^{(i-1)}$ with the column vector $\delta^{B_{i-1}}$.
            \STATE Compute $\delta^{B_i}$ by applying Algorithm~\ref{alg:description} with inputs $X^B$ and $V^{(i)}$ (substituting for $X^A$).
            \ENDFOR
\end{algorithmic}
\end{algorithm}

\begin{figure}[tb]
   \centering 
   \includegraphics[width=0.96\linewidth]{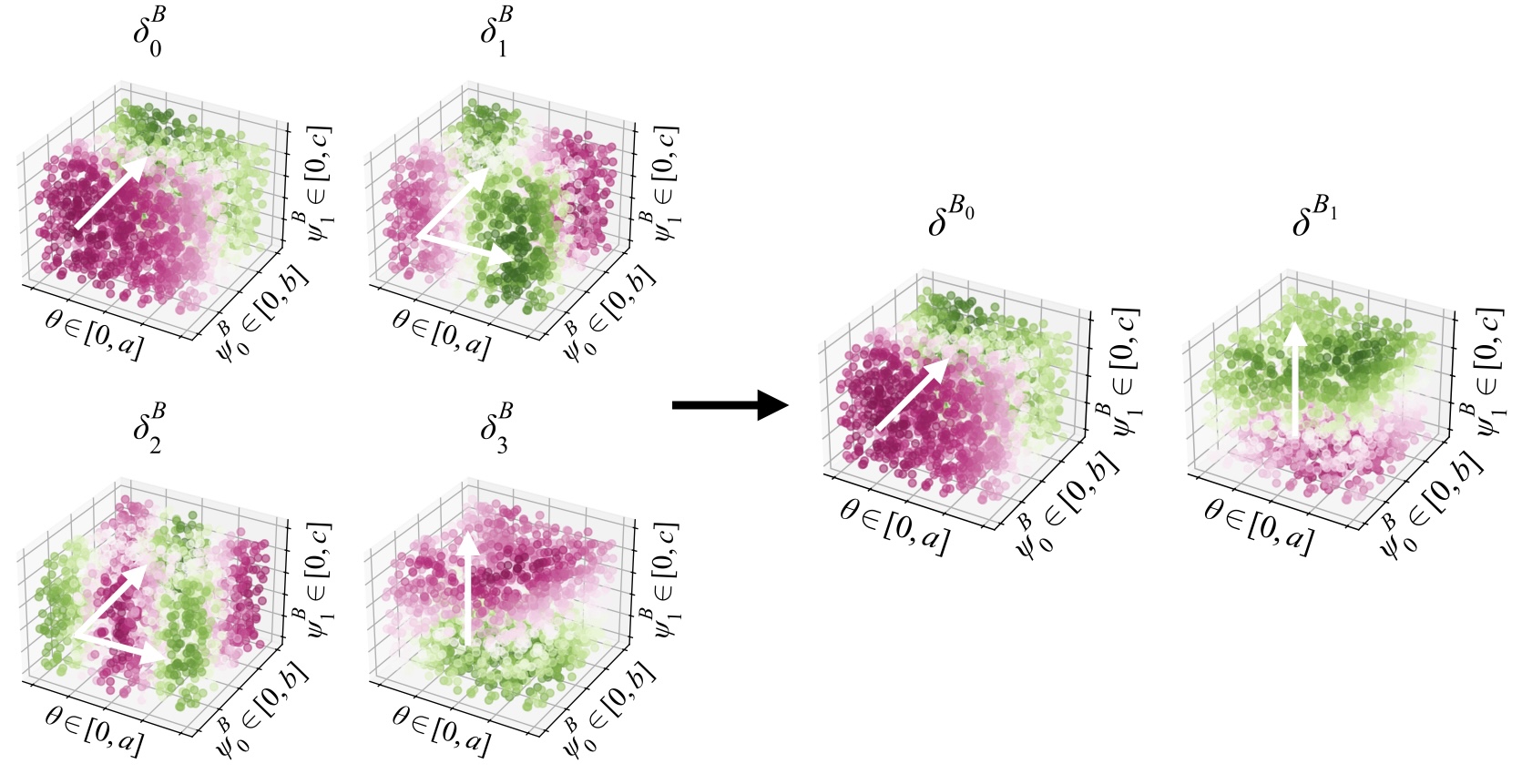}
   \caption{
   Line $(X^A)$ vs. cube $(X^B)$ example. The $\theta$ coordinate is shared between the two datasets. The second and third coordinates $\psi^B_1,\psi^B_2$ are unique to $X^B$. Each point in the scatter plots presents an observation in $X^B$. \textbf{Left panel:} The observations in $X^B$ are colored according to the values of its corresponding entries in the leading differential vectors. \textbf{Right panel:} The observations in $X^B$ colored according to the values of its corresponding entries in the iterative differential vectors $\delta^{B_0}$ and  $\delta^{B_1}$. 
   In all figures, the white arrow points to the direction of the dominant parameter.}
   \label{fig: eigenvectors of L^A and L^B 3D cube}
\end{figure}


\section{Analysis}\label{Sec: Analysis}

We begin with two sections that provide background for the analysis. Section \ref{sec:convergence} reviews existing results on the convergence of Laplacian eigenvectors under the manifold model. Section \ref{sec:product_manifold} introduces the product manifold model and establishes notation and basic properties. In addition, it states the assumptions underlying our analysis of the multi-modal setting. In Section \ref{sec:convergence_delta}, we present the convergence results and a proof outline. The full proof is deferred to the appendix.

\subsection{Convergence of the graph Laplacian eigenvectors}\label{sec:convergence}
In recent years, several works have derived convergence results, under the manifold setting, of the eigenvectors of the discrete Laplacian to eigenfunctions of the Laplace-Beltrami (LB) operator \citep{von2008consistency, singer2017spectral, garcia2020error,wormell2021spectral,dunson2021spectral, calder2022improved}.
For example, \cite{cheng2022eigen} derived an $l_2$ norm convergence result for the normalized and unnormalized Laplacian matrices,  under the following assumptions:
\begin{enumerate}[label=(\roman*)]
    \item The $n$ observations are drawn independently according to a uniform distribution over a $d$-dimensional manifold $\M$.
    \item The smallest $K$ eigenvalues of the Laplace-Beltrami (LB) operator over $\M$ have single multiplicity, with a minimal spectral gap $\gamma_K>0$.
    \item The graph weights are computed using a Gaussian kernel with a bandwidth parameter $\epsilon_n \to 0^{+}$ that satisfies $\epsilon_n^{d/2+2} > C_K \frac{\log n}{n}$ for some constant $C_K$.
\end{enumerate}
Let $dV$ denote the Riemannian volume form of $\M$ and let the constant
$p = 1/\int_\M dV$ denote the uniform sampling density.
We denote by $\eta_k$ the $k$\textsuperscript{th} smallest eigenvalue of the Laplace-Beltrami operator and by $f_k$ a corresponding normalized eigenfunction, where
\begin{align}
    \int_\M f_k(x)^2 p dV(x) = 1. 
\end{align}
We denote by $\rho_X (f_k) \in \R^n$ the operator that evaluates the function $f_k$ at a point set $X = \{x_1, \ldots, x_n \} \subseteq \M$, and by $\phi_k(X) \in \R^n$ the vector of samples,
\begin{align}
    \label{eq: vector of samples}
    \phi_k(X) = \frac{1}{\sqrt{p n}}\rho_X (f_k).
\end{align}
In addition, we denote by $v_k \in \R^n$ the eigenvector of the symmetric-normalized graph Laplacian that corresponds to the $k$-th smallest eigenvalue. Similarly, we denote by $v^{(rw)}_k$ the eigenvectors of the random walk Laplacian  $L^{(rw)} = I - D^{-1}W$.
The following convergence guarantee for $v^{(rw)}_k$ was proved by Cheng and Wu: 
\begin{theorem}[Theorem 5.5, \cite{cheng2022eigen}]\label{thm:convergence_rw}
For $n \to \infty$ and under assumptions (i)-(iii), w.p. $>1-4K^2n^{-10} - (4K+6)n^{-9}$, the k\textsuperscript{th}  eigenvector of the random-walk Laplacian satisfies
\begin{equation}\label{eq:eigenvectors_product}
   \|v_k^{(rw)} - \alpha \phi(X)_k\|_2
   =
   \OO\left(\epsilon_n + \epsilon_n^{-d/4-1/2}\sqrt{\log n/n}\right),
\qquad \forall k \leq K,
\end{equation}
where $v_k^{(rw)}$ is $D$-normalized such that $(v_k^{(rw)})^T D v_k^{(rw)}/(n p) = 1$, $\epsilon_n$  is the bandwidth parameter used for the Gaussian affinity, equal to $\sigma^2/2$ in our notation \eqref{eqn: gaussian kernel weights},
and $|\alpha| = 1+o_p(1)$. 
\end{theorem}
For our analysis, the convergence rate of $\alpha$ must be made explicit.
Lemma \ref{lem:alpha_convergence} in the supplementary materials shows that $|\alpha|-1$ converges to zero at the same rate as in Eq. \eqref{eq:eigenvectors_product}. 
Based on this theorem and a concentration result for the degree matrix, a similar concentration bound follows for the eigenvectors of the symmetric normalized Laplacian matrix:
\begin{theorem}\label{thm:convergence_symmetric_normalized}
For $n \to \infty$, under assumptions (i)-(iii),  w.p. $>1-4K^2n^{-10} - (4K+8)n^{-9}$, the $k$\textsuperscript{th} eigenvector $v_k$ of the symmetric-normalized Laplacian satisfies
\begin{equation} \label{eq:vksym_conv}
    \|v_k - \alpha \phi_k(X)\|_2
    =
    \OO\left(\epsilon_n + \epsilon_n^{-d/4-1/2}\sqrt{\log n/n}\right),
\qquad \forall k \leq K,    
\end{equation}
where $\|v_k\|=1$ is normalized and $|\alpha|=1+o_p(1)$.
\end{theorem}
The proof is in Section~\ref{sec:convergence_proof_vksymm}.
Theorems \ref{thm:convergence_rw} and \ref{thm:convergence_symmetric_normalized} are concerned with the convergence of graph eigenvectors computed using a single dataset. However, in this paper, we are interested in the case of two datasets, assuming that each was generated by random sampling from a product manifold. We start by providing the necessary notations and definitions for this case.

\subsection{The product manifold model}\label{sec:product_manifold}
Let $\M_1,\M_2$ be two manifolds, and let $\M = \M_1 \times \M_2$ denote the product manifold, such that every point $x \in \M$ is associated with a pair of points $x_1 \in \M_1$ and $ x_2 \in \M_2$. We denote by $\pi_1: \M \to \M_1$ and $\pi_2: \M \to \M_2$ the canonical projections of a point in $\M$ to its corresponding points in $\M_1,\M_2$, respectively. 
Let $f_1: \M_1 \to \mathbb{R}$ be a real function on the manifold $\M_1$. The canonical projection operator can be used to define a function $f: \M \to \mathbb{R}$ known as the pullback of $f_1$,
\begin{align}
    f(x) = f_1(\pi_1(x)).
\end{align}
In this work, we consider three disjoint sets of latent variables: $\psi^A$ (latent variables unique to $A$), $\psi^B$ (latent variables unique to $B$), and $\theta$ (common latent variables). Formally, we assume the existence of three latent manifolds $\M_1,\M_2,\M_3$ that are smooth transformations of the three sets of variables, $\psi^A,\psi^B,\theta$, respectively. 
The two sets of observations $X^A, X^B$ are sampled, respectively, from two product manifolds $\M_A$ and $\M_B$, where
\begin{align}
    \M_A = \M_1 \times \M_3, \qquad \M_B = \M_2 \times \M_3.
\end{align}
We assume that the observations in $X^A$ and $X^B$ are generated according to the following steps: (i) For each $i$, the latent variables $\theta_i, \psi^A_i, \psi^B_i$ are drawn independently.
(ii) The observations $x_i^A, x_i^B$ are computed according to 
\begin{align}
    (\theta_i, \psi_i^A) \xrightarrow{T_A} x_i^A, \qquad (\theta_i, \psi_i^B) \xrightarrow{T_B} x_i^B,
\end{align}
where $T_A$ and $T_B$ are smooth maps.
We denote by $f^{(j)}_i: \M_j \to \mathbb{R}$ an $i$\textsuperscript{th} eigenfunction of the Laplace-Beltrami operator of the manifold $\M_j$.
The eigenfunctions of the products $\M_A, \M_B$ are equal, up to a constant,  to the pointwise product of the eigenfunctions of $\M_1,\M_2,\M_3$. See Theorem 2 in \cite{zhang2021product} and the discussion on warped product manifolds in \cite{pmlr-v221-he23a}. That is, for any pair  of functions $f^{(j)}_\ell$ and $f^{(j')}_k$ that are eigenfunctions, respectively, of $\M_j$ and $\M_{j'}$, the function $f^{(j)}_\ell(\pi_{j}(x))f^{(j')}_k(\pi_{j'}(x))$ is an eigenfunction of the product manifold $\M_j \times \M_{j'}$. Thus, one can index the eigenfunctions of the product manifold according to the index of the corresponding components $(k,l)$.
\begin{equation}\label{eq:eigenfunctions_product}
f^{A}_{l,k}(x) \sim f^{(1)}_l(\pi_1(x)) \cdot f^{(3)}_k(\pi_3(x)), \qquad f^{B}_{m,k'}(x) \sim f^{(2)}_m(\pi_2(x)) \cdot f^{(3)}_{k'}(\pi_3(x)).
\end{equation}
For the three manifolds $\{\M_j\}_{j=1}^3$, we denote by $\eta^{(j)}_i$ the $i$\textsuperscript{th} smallest eigenvalue, that corresponds to the eigenfunction $f_i^{(j)}$.
We denote by $\eta^A_{l,k},\eta^B_{l,k}$ the $(l,k)$-th smallest eigenvalue of the products $\M_A,\M_B$, as ordered by the following sums,
\begin{equation}\label{eq:eigenvalues_product}
\eta^{A}_{l,k} = \eta^{(1)}_l + \eta^{(3)}_k \qquad \qquad\eta^{B}_{m,k'} = \eta^{(2)}_m + \eta^{(3)}_{k'}\,.
\end{equation}
The $m$\textsuperscript{th} smallest eigenvalue has the subscript $(l,k)$ if $\eta^A_{l,k}$ is the $m$\textsuperscript{th} smallest eigenvalue of $\M_A$. 
We denote by $v^A_{l,k}, v^B_{m,k'}$ the eigenvectors of the symmetric normalized Laplacian matrices $L^A,L^B$.
Similarly to the case of a single manifold, we denote by $\rho_X(f^A)$ the sampling operator that evaluates a function $f^A$ at the points of $X^A$. In addition, 
we denote by $\rho_{\pi_i(X)}(f^{(i)})$ the sampling operator that takes a  function $f^{(i)}$ on $\M_i$ and evaluates it at the points $\pi_i(x_1), \ldots \pi_i(x_n)$.
We denote by $p_A,p_B$ the sampling density over the manifolds $\M_A,\M_B$. 
\subsection{Convergence of differential vectors}\label{sec:convergence_delta}
Consider the eigenvectors of $L^A$ and $L^B$, denoted by $v_{l,k}^A$ and $v_{l,k}^B$, respectively. The following lemma, which forms the basis of our analysis, shows that $v_{l,k}^A$ and $v_{m,k'}^B$ are nearly orthogonal.
\begin{lemma}\label{lem:concentration} Let $\M_A = \M_1 \times \M_3$ and $\M_B = \M_2 \times \M_3$. Let $v^A_{l,k},v^B_{m,k'}$ be, respectively, the $(l,k)$-th and $(m,k')$-th unit-length eigenvectors of the Laplacian matrices $L^A,L^B$.
We assume that the corresponding eigenvalues $\eta^A_{l,k}$ and $\eta^B_{m,k'}$ are both within the $K$ smallest eigenvalues of their respective spectra.
Under assumptions (i)-(iii), as $n \to \infty$ the following holds  with probability $>1-12K^2 n^{-10} - (8K+16)n^{-9}$, 
\begin{equation}\label{eqn:orthogonality of eigenvectors}
    |(v^A_{l,k})^Tv^B_{m,k'}| = \begin{cases}
    1 -\OO\bigg( \epsilon_n + \sqrt{\frac{\log n}{n \epsilon_n^{d/2+1} }}\bigg), & \text{if $k=k'$ and $l=m=0$}.\\
    \OO\bigg( \epsilon_n + \sqrt{\frac{\log n}{n \epsilon_n^{d/2+1} }}\bigg), & \text{otherwise}.
  \end{cases}   
  \end{equation}
\end{lemma}
Lemma \ref{lem:concentration} proves that the eigenvectors of the Laplacian matrices, excluding those that are associated only with the shared variable $\theta$, are almost orthogonal. The lemma is the main building block for the proof of our convergence guarantee derived for  Algorithm~\ref{alg:description}. 
For simplicity, in our proof, we set the filter $H(L^A)$ to be a threshold function with parameter $\tau$ as defined in Eq. \eqref{eq:threshold_filter}. 
Given a matrix $V^A$ of size $n\times K_{A}$ whose columns consist of the eigenvectors of $L^A$ with eigenvalues smaller than $\tau$, the filter matrix $H(L^A)$ is equal to the following projection matrix, which we denote by $ Q^A=I-V^A (V^A)^T$. 
In addition, for the analysis, we make one modification to Algorithm~\ref{alg:description}. 
In step 4 of the algorithm, we apply a low-pass filter to $P^B$ and attain $P^B_\tau$ given by
\begin{align}
    P^B_\tau = \sum_{l,k; \lambda^B_{lk} \leq \tau} (1-\lambda^B_{l,k}) v_{l,k}^B (v_{l,k}^B)^T.
\end{align}
Thus, the operator computed in step 4 of the algorithm is equal to
$
\tilde P_A = Q^A P^B_\tau Q^A.
$
We denote by $K_B$ the number of eigenvectors in $L_B$ whose eigenvalue is smaller than $\tau$, and by $K = \max(K_A, K_B)$.
Let $\eta_{1,0}^B$ be the first eigenvalue whose corresponding eigenfunction is associated with a differential variable. We assume that the chosen threshold is sufficiently large such that $\tau > \eta_{1,0}^B$. 
The following theorem gives the convergence rate of the leading differential vector $\delta^B$ to the eigenfunction $f_1^{(2)}$ of $\M_2$.


\begin{theorem}\label{thm:main_theorem}
Let $X^A, X^B$ be two sets of $n$ points sampled uniformly at random from the product manifold $\M_A = \M_1 \times \M_3, \M_B = \M_2 \times \M_3$ respectively.
Let
\begin{align}
    \delta^B
    =
    \argmax_{\|v\|=1}\ v^T Q^A P^B_\tau Q^A v
\end{align}
be the differential vector obtained in step 5 of Algorithm~\ref{alg:description}. 
For $n \to \infty$ and assuming (i)-(iii) hold for both manifolds, the following holds  with probability $1-4K^2n^{-10} - (2K+6)n^{-9}$,
\begin{equation}\label{eqn: Theorem 2}
\Big\| \delta^B - \frac{\alpha}{ \sqrt{p n}}\rho_{\pi_2(X)}f^{(2)}_1 \Big \|^2 \leq \OO( K\epsilon_n) + \OO\bigg( K\sqrt{\frac{\log n}{n \epsilon_n^{d/2+1} }}\bigg).
\end{equation}
\end{theorem}
Theorem \ref{thm:main_theorem} implies that given a sufficient number of samples, the leading vector $\delta^B$ of the filtered operator $\tilde P^B$ computed in step 5 of Algorithm~\ref{alg:description} converges to the leading eigenfunction of $\M_2$. Thus, it captures the leading differential variable that is not shared between the two datasets. Note, however, that the bound is on the squared distance, which implies that the rate of convergence is slower than the rate in Theorem \ref{thm:convergence_symmetric_normalized}.

\paragraph{Proof sketch for Theorem \ref{thm:main_theorem}}
Recall that $\eta^{A}_{l,k} = \eta^{(1)}_l + \eta^{(3)}_k$ denotes the eigenvalues of the Laplace-Beltrami operator on the product manifold $\M_A = \M_1 \times \M_3$. The eigenvalues $\eta_{0,k}$ correspond to eigenfunctions $f_{0,k}$ that are functions of the shared variable $\theta$ only (see Eq.~\eqref{eq:eigenfunctions_product}). 
For the analysis, we partition the columns of $V^A$, which contain the eigenvectors of $L^A$ with smallest eigenvalues, to two parts: (i) $V^{A_\theta}$, which includes only the eigenvectors of the form $\{v_{0,k}\}$ which correspond to eigenfunctions $f_{0,k}$, and (ii) $V^{A_\psi}$ which includes vectors $\{v_{l,k}\}$ where $l \neq 0$. In addition, we define the following projection matrices:
\begin{align}
    Q^{A_\theta} = I - V^{A_\theta} (V^{A_\theta})^T,
    \qquad 
    Q^{A_\psi} = I - V^{A_\psi} (V^{A_\psi})^ T.
\end{align}
Note that due to the orthogonality of $V^{A_\psi}$ and $V^{A_\theta}$, we have
\begin{align}
    Q^{A_\theta} Q^{A_\psi} = I - V^A (V^{A})^T.
\end{align}
Similarly, we denote by $V^{B_\theta}$ a matrix that contains, as columns, the eigenvectors of $L^B$ that correspond to $\lambda^B_{0,k'}$, and by $Q^{B_\theta} = I- V^{B_\theta}(V^{B_\theta})^T$. 
We prove Theorem \ref{thm:main_theorem} in three steps:
\begin{itemize}[leftmargin= .6in]
    \item[\textbf{Step 1:}] Let $E_1 = Q^{B_\theta} P^B_\tau Q^{B_\theta}$. The projection matrix $Q^{B_\theta}$ is the \textit{ideal filter} in the sense that it perfectly removes the leading eigenvectors that are functions of the shared variable. We show that the leading eigenvector of $E_1$ converges, for large $n$, to $\rho_X(f_{1,0}^B)$.
    \item[\textbf{Step 2:}] Let $E_2 = Q^{A_\psi}Q^{A_\theta} P^B_\tau Q^{A_\theta}Q^{A_\psi}$. 
    This is the filtered operator $\tilde P^B$, computed in step $4$ of Algorithm \ref{alg:description}.
    We bound the spectral norm of the difference of $E_1$ and $E_2$.
    \item[\textbf{Step 3:}] Using the Davis-Kahan theorem, we bound the spectral norm of $E_1-E_2$ and conclude that the leading eigenvector of $E_2$ converges to $\rho_X(f^B_{1,0})$.
\end{itemize}
A detailed proof of each step appears in Section \ref{sec:three_part_proof} of the supplementary materials.

\section{Computational complexity}\label{sec:complexity}
A naive implementation of Algorithm~\ref{alg:description} and \ref{alg:description 2} requires $\OO(n^3)$ operations for computing all eigenvalues and eigenvectors of the Laplacian matrices $L_A$ and $L_B$ and the differential operators, a complexity which may be prohibitive even for datasets of moderate size. However, one may reduce the complexity to $\OO(n^2)$ by approximating the filtered Laplacian $H(L_A)$ with the leading eigenvectors of $L_A$. 
Let us consider the steps of Algorithm~\ref{alg:description} for computing the differential vector $\delta^A$.
\begin{itemize}
    \item[1-2.] The complexity of computing the weight matrices $W_A,W_B$ and the Laplacian matrices $L_A,L_B$ is $\OO(n^2)$. 
    \item[3] Computing the high-pass filter matrix $H(L_A)$ defined by the step function in Eq. \eqref{eq:threshold_filter} requires $\OO(n^2  C_\tau)$, where $C_\tau$ is the number of eigenvectors with eigenvalue smaller than $\tau$.  
    For a general filter function $h(\lambda)$, one can approximate $H(L_A)$ via 
    \begin{align}\label{eq:HLA_approx}
    H(L_A) &= \sum_i h(\lambda_i) v_i^A (v_i^A)^T = I - \sum_{i=1}^n (1-h(\lambda_i)) v_i^A (v_i^A)^T \notag \\
    &\approx I - \sum_{i=1}^{C_\tau} (1-h(\lambda_i)) v_i^A (v_i^A)^T,  
    \end{align}
    where the threshold $\tau$ is set according to the spectra of $L_A$.
    \item[4-5] By using the approximation of $H(L_A)$ in Eq. \eqref{eq:HLA_approx}, computing $H(L_A) P_B H(L_A)$ requires only $\OO(C_\tau n^2)$ operations, and the complexity of computing its leading eigenvector $\delta_A$ is $\OO(n^2)$.
    
\end{itemize}
Thus, the complexity of all steps in Algorithm~\ref{alg:description} is quadratic in the number of samples. A similar derivation can also be done for Algorithm~\ref{alg:description 2}. 

\section{Simulations}\label{Sec: Examples} We evaluate our method on datasets of different types. For each experiment, we compare the outcome of our approach to the following two baselines: the \citet{fukunaga1970application} Transform (FKT)  applied to the graph Laplacian, and the approach developed by \cite{shnitzer2019recovering} (referred to as Shnitzer +). Details regarding these methods are provided in \ref{app: methods}.
Additional experiments on simulated and real ECG measurements are provided in~\ref{app:additional_exp}.

\subsection{Rectangle vs. Line}
\paragraph*{Data} We present results from the line vs.\ rectangle experiment described in Section~\ref{Sec: Algorithm - one variable}, where 
$X_B$ is sampled uniformly from the rectangle $[0,a]\times[0,b]$, while $X_A$ consists only of the first coordinate in $[0,a]$.

\paragraph*{Goal}
To assess the recovery of a latent one-dimensional variable in a simple product-manifold setting and to illustrate the convergence behavior of the proposed methods.

\paragraph*{Evaluation Metric} 
Let $\psi^B(x) = \cos(\pi x/b)$ be the first non-trivial eigenfunction of the one-dimensional Laplacian with Neumann boundary conditions on $[0, b]$, and let $\psi^B_i$ be the discretization  of $\psi^B$ at points equal to the second coordinate of $X^B$ such that
\[
\psi^B_i = \cos(\pi x^B_{i,1} / b) \qquad i = 1,\ldots,n
\]
Recovery of $\psi^B$ is assessed by the absolute value of the Pearson correlation between $\delta^B_0$ and $\psi^B_{1:n} = (\psi^B_1, \ldots, \psi^B_n)$.


\paragraph*{Results} Table~\ref{tab:rectangle_results} presents the mean correlation and standard deviation for each method, based on $500$ repetitions. 
As indicated in the table, DELVE consistently 
recovered the hidden parameter. 
The FKT method also delivered strong results, with a slightly lower mean correlation than DELVE. In contrast, neither the real nor the imaginary components of the leading eigenvector of operator $A$ in \citet{shnitzer2019recovering} effectively recovered $\psi^B$ under this setting.

Next, we compare the empirical convergence rates of the differential vectors with the theoretical results from Section~\ref{Sec: Analysis}. Let
\[
\theta_i = \cos\!\left(\pi x_i^A/a\right), \qquad i=1,\ldots,n.
\]
Our analysis builds on Theorem~5.5 in~\cite{cheng2022eigen}, which establishes eigenvector convergence for the random walk Laplacian computed from a single dataset. 
Figure~\ref{fig:convergence_line_rectangle} shows $\|v^{(rw)} - \theta_{1:n}\|_2$ (blue) and $\|\delta^B - \psi^B_{1:n}\|_2$ as functions of the sample size (log--log scale). Following the procedure in Section~7.1 of~\cite {cheng2022eigen}, for each $n$ we tested seven candidate bandwidth values for the kernel in Eq.~\eqref{eqn: gaussian kernel weights} and selected the best one. 

The empirical convergence rate of $v^{(rw)}$ is approximately $-0.42$, which is faster than the theoretical bound in~\cite{cheng2022eigen}, though slightly slower than their reported empirical rate of $-0.49$. This discrepancy likely stems from the specific bandwidth choice. Consistent with Theorem~\ref{thm:main_theorem}, the convergence rate of $\delta^B$ is approximately half that of $v^{(rw)}$, namely $-0.216$.

\begin{table}[tb]
\centering
\resizebox{\textwidth}{!}{
\begin{tabular}{lllll}
\toprule
 & DELVE & Shnitzer + (real) & Shnitzer + (imag) & FKT \\
\midrule
Mean (SD) & \textbf{0.973 (0.001)} & 0.118 (0.089) & 0.077 (0.056) & 0.901 (0.005) \\
\bottomrule
\end{tabular}

}
\caption{\textbf{Rectangle vs. Line}. Absolute value of the Pearson correlation for 500 independent runs,  between differential vectors obtained by three methods and $\psi^B(x) = \cos(\pi x/b)$. 
}
\label{tab:rectangle_results}
\end{table}

\begin{figure}[tb]
\centering
   \includegraphics[width=0.7\linewidth]{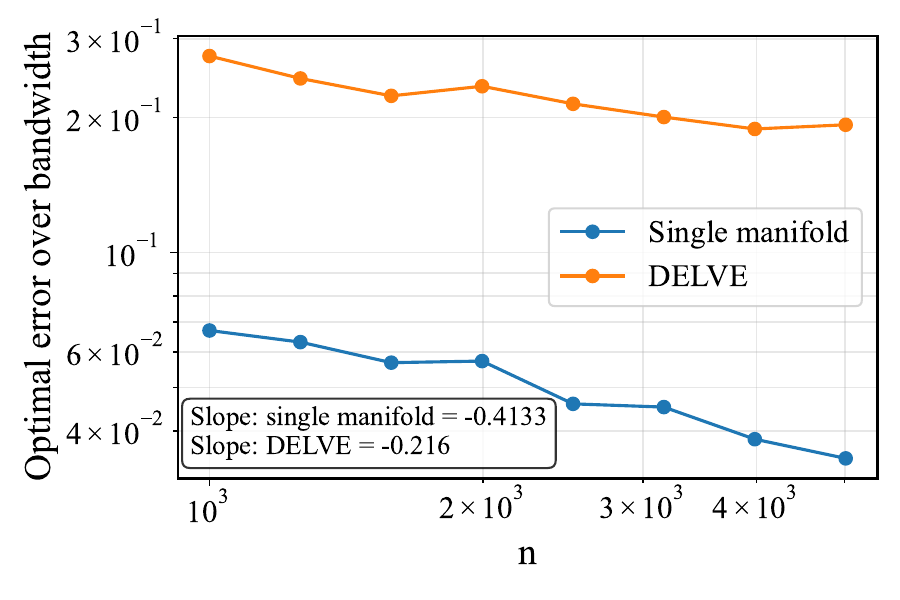}
   \caption{\textbf{Rectangle vs. Line}. Blue curve: difference in $l_2$ norm between leading Laplacian eigenvector $v^{(rw)}$, computed from $x_i^A$ and the theoretical vector whose elements equal $\cos(\pi x_i^A/a)$. Orange curve: difference between 
   differential vector $\delta^B$ and $\cos(\pi x_{i,1}^B/b)$. For both curves, the difference is presented as a function of the number of samples $n$ on a log-log scale.} 
   \label{fig:convergence_line_rectangle}
\end{figure}

\subsection{Synthetic Torus Data}
\paragraph*{Data} We generated $n=2,000$ triplets $\theta, \psi^A, \psi^B \sim \text{Uniform}[0,2\pi ]$. From these, we constructed two sets of points in $\mathbb{R}^3$, each forming a 2-dimensional torus:
\[
\mathbf{X}^{A} = 
\begin{pmatrix} 
(R + r^{(1)} \cos( \psi^A)) \cos( \theta) \\
(R + r^{(1)} \cos( \psi^A)) \sin( \theta) \\
r^{(1)} \sin (\psi^A)
\end{pmatrix}, 
\quad
\mathbf{X}^{B} = 
\begin{pmatrix} 
(R + r^{(2)} \cos(\psi^B)) \cos( \theta) \\
(R + r^{(2)} \cos(\psi^B)) \sin(\theta) \\
r^{(2)} \sin(\psi^B)
\end{pmatrix}.
\]
Both tori share the same "primary" angle $\theta$, but differ in their "secondary" angles $\psi^A$ and $\psi^B$. This makes the shared structure dominant while the differences are subtle. The parameters in our experiment were set to $R = 10$, $r^{(1)} = 4$, and $r^{(2)} = 2$.

\begin{figure}\hspace{-14pt}
   \includegraphics[width=13.5cm]{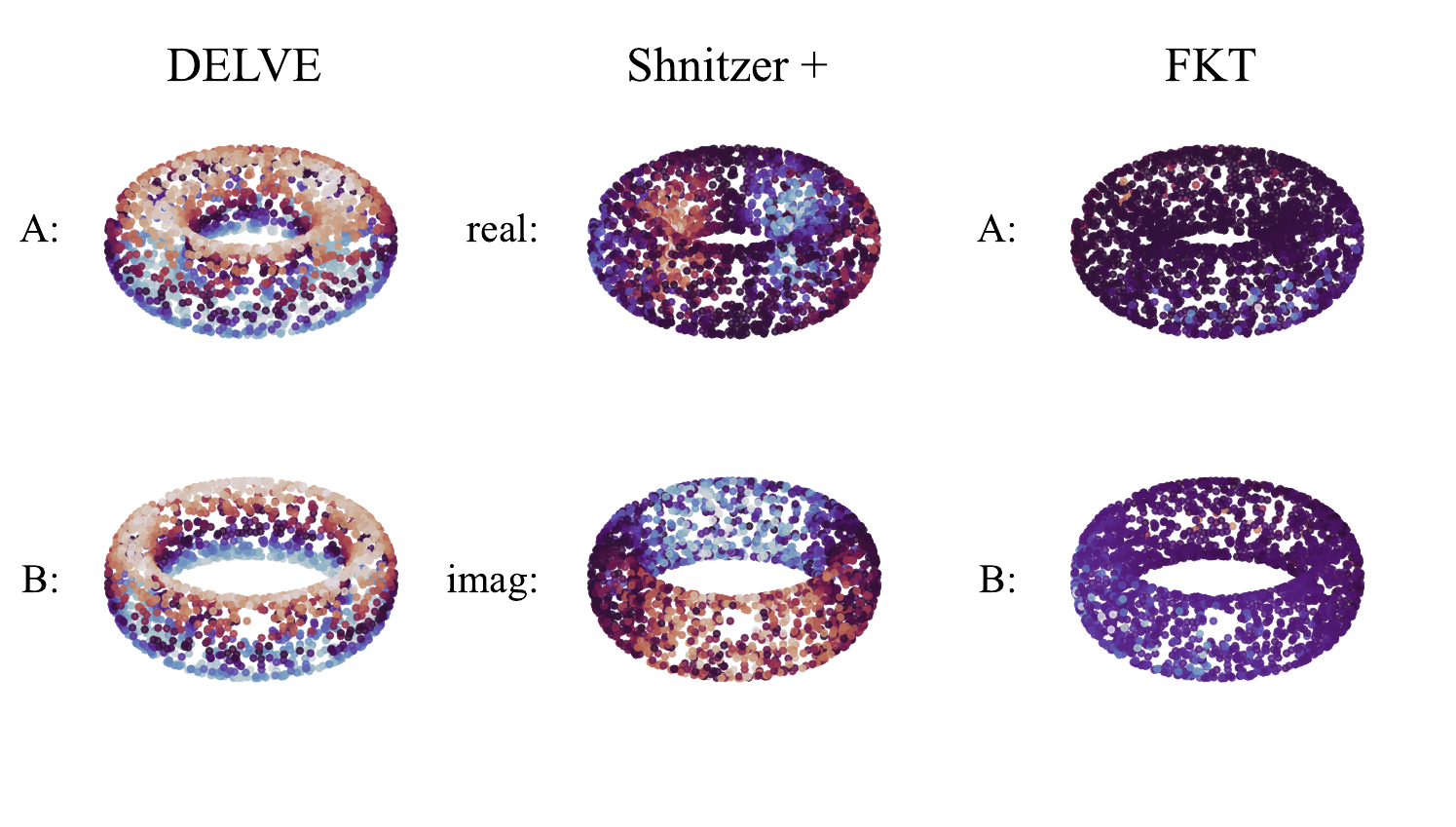}
   \caption{\textbf{Multimodal Tori}. \textbf{Top:}
   $X^A$ colored by the leading vector associated with $\psi^A$ in each method. \textbf{Bottom:}
   $X^B$ colored by the leading vector associated with $\psi^B$.}
   \label{fig:torus_vis}
\end{figure}

\paragraph*{Evaluation Metric} Due to rotational symmetry, in this setting, the recovered vector has an arbitrary phase-shift with respect to the latent parameters $\psi^A$ and $\psi^B$. To account for this, we evaluate recovery using a cyclic Pearson correlation that measures alignment with $\cos(\psi)$, making the metric invariant to phase shifts.

\paragraph*{Results} As shown in Table~\ref{tab:torus_results}, DELVE’s $\delta_0^A$ and $\delta_0^B$ are highly correlated ($r>0.99$) with the true latent parameters $\cos(\psi^A)$ and $\cos(\psi^B)$, whereas Shnitzer et al. (2019) and FKT yield only weak correlations.
Figure~\ref{fig:torus_vis} shows that, when observations are colored by the recovered vectors, DELVE highlights the inner-circle variation corresponding to the latent component, while coloring by the Shnitzer et al. (2019) vectors emphasizes the shared angle $\theta$. Conversely, FKT displays no coherent association and exhibits pronounced outliers despite the absence of outliers in the data.

\begin{table}[tb]
\centering
\resizebox{\textwidth}{!}{
\begin{tabular}{lllll}
\toprule
 & DELVE & Shnitzer + (real) & Shnitzer + (imag) & FKT \\
\midrule
$\psi^A$: Mean (SD) & \textbf{0.991 (0.003)} & 0.014 (0.028) & 0.026 (0.013) & 0.023 (0.013) \\
$\psi^B$: Mean (SD) & \textbf{0.996 (0.003)} & 0.015 (0.031) & 0.018 (0.010) & 0.023 (0.013) \\
\bottomrule
\end{tabular}

}
\caption{\textbf{Multimodal Tori}. Circular Pearson correlation (CPC), averaged over 500 runs,  between differential vectors obtained by three methods and $cos(\psi^A), cos(\psi^B)$ (higher is better).} 
\label{tab:torus_results}
\end{table}

\subsection{Rotating Dolls}\label{sec: yoda example}
\paragraph*{Data} The datasets $X^A, X^B$ consist of 4050 images of rotating dolls, created by \cite{lederman2018learning}\footnote{ Available online at \url{https://roy.lederman.name/alternating-diffusion/}}. In this example, three dolls were placed on rotating displays with varying rotation rates. Simultaneous snapshots of the rotating dolls were captured by two cameras: the first camera recorded the rabbit and bulldog dolls, while the second camera captured the bulldog and Yoda dolls. Under this setting, the rotation angle of the bulldog is the shared latent variable, $\theta$, and the rotation angles of the Yoda and the rabbit are the unique latent variables, $\psi^A, \psi^B$, respectively.

\paragraph*{Evaluation metric}

To obtain ground-truth rotation angles, each image was split into two halves, each containing a single figure. For each half, a Laplacian graph was constructed, and the two non-trivial leading eigenvectors of that graph were used as proxies for the true rotation angles. For example, for the cropped bulldog images, the two eigenvectors $v^{\text{bulldog}}_1, v^{\text{bulldog}}_2$ were used to compute the estimated angle as $\hat{\theta} = \arctan(v^{\text{bulldog}}_1,v^{\text{bulldog}}_2)$. Performance was quantified by computing the absolute value of the Pearson correlation between the vectors produced by each method and the corresponding sine of the ground-truth angles.

\paragraph*{Results}
Table~\ref{tab:Yoda_results} shows that DELVE achieves high correlation with both $\psi^A$ and $\psi^B$. The FKT method performs comparably in this setting. In contrast, the real and imaginary components of the leading vector of operator $A$ proposed by \citet{shnitzer2019recovering} show substantially lower correlation with both ground-truth angles.

\begin{table}[tb]

\centering
\begin{tabular}{lrrrr}
\toprule
 & DELVE & Shnitzer + (real) & Shnitzer + (imag) & FKT \\
\midrule
$\psi^B$ (rabbit) & \textbf{0.928} & 0.022 & 0.000 &\textbf{0.922} \\
$\psi^A$ (Yoda) & \textbf{0.995} & 0.011 & 0.057 & \textbf{0.995}\\
\bottomrule
\end{tabular}

\caption{\textbf{Rotating Dolls}. Absolute value of the Pearson correlation between the differential vectors obtained by three methods and $sin(\psi^A), sin(\psi^B)$ (higher is better). }
\label{tab:Yoda_results}
\end{table}


\subsection{Accelerometer Sensors}

\paragraph*{Data}
The \textit{Human Activity Recognition (HAR) dataset} \citep{human_activity_recognition_using_smartphones_240, anguita2013public} contains recordings of human activities collected from smartphone sensors, with $n=10,299$ observations across six activity labels: (i) walk, (ii) walk upstairs, (iii) walk downstairs, (iv) stand, (v) sit, and (vi) lay.
We focus on accelerometer-based measurements and construct two sensor-specific feature sets. 
The first dataset, $X^A$, consists of time-domain features extracted from the body accelerometer, capturing dynamic motion components with gravity removed.
The second dataset, $X^B$, consists of time-domain features extracted from the gravity accelerometer, capturing orientation and posture-related information. 

\paragraph{Goal} 
Comparing posture-dominated and motion-dominated sensing modalities, we demonstrate that sensor-specific latent variables encode meaningful task-related structure. This experiment indicates that important information is contained not only in shared representations but also in modality-specific components.


\paragraph{Evaluation metric}
In addition to the differential vectors computed for each method, we also computed the shared component of the two modalities. For DELVE, this was done according to Eq.~\eqref{eqn:projection of LA onto LB}. The shared component of FKT was computed similarly. The computation of a shared component is also derived in \cite{shnitzer2019recovering}. 
For each method, we applied K-means clustering to four embeddings: (i-ii) two differential vectors of modality $A$ and $B$, (iii) two vectors that constitute the shared component, and (iv) a concatenation of all three. 
Clustering performance was quantified with the Adjusted Rand Index (ARI) and Normalized Mutual Information (NMI). 


\paragraph{Results}
Figure~\ref{fig: Acc-gyr uniques} shows the embeddings obtained from each modality alone (top) and after augmenting them with the differential component learned from the other sensor (bottom). The single-sensor embeddings provide only coarse separation: for example, gravity alone primarily distinguishes lay from non-lay poses. Incorporating the body-accelerometer–specific component further separates walking from the static activities and reveals structure that would otherwise remain collapsed.

\begin{figure}[htb!]
   \includegraphics[width=0.99\linewidth]{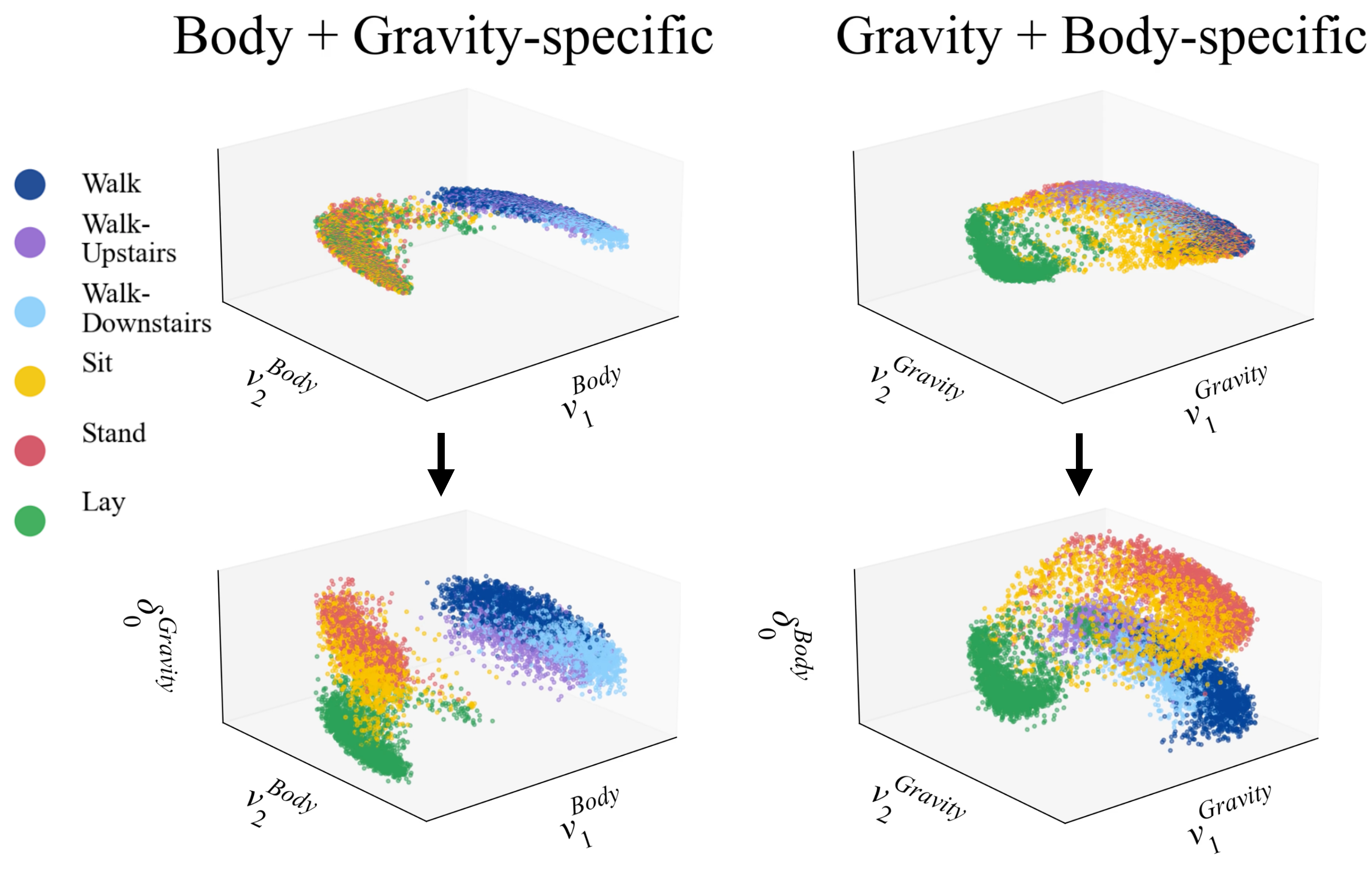}
    \caption{\textbf{Accelerometer Sensors.} Spectral embeddings of the body- and gravity-accelerometer modalities. \textbf{Top:} Two-dimensional embeddings obtained from the two leading non-trivial eigenvectors of the random-walk Laplacian, computed for each modality separately. \textbf{Bottom:} Augments these embeddings with a third coordinate, the learned differential vector from the complementary modality, thereby exposing the additional modality-specific information contributed by each sensor. Colors denote activity labels. }

   \label{fig: Acc-gyr uniques}
\end{figure}

Figure~\ref{fig: Acc-gyr uniqe-vs-shared} illustrates the effect of combining shared and sensor-specific DELVE vectors. The UMAP embedding of the two leading shared vectors merges all walking observations into a single cluster, whereas adding the two leading differential vectors separates regular walking from other walking activities. This demonstrates that modality-specific variation provides discriminative information beyond the shared representation.

\begin{figure}[htb!]
   \includegraphics[width=0.99\linewidth]{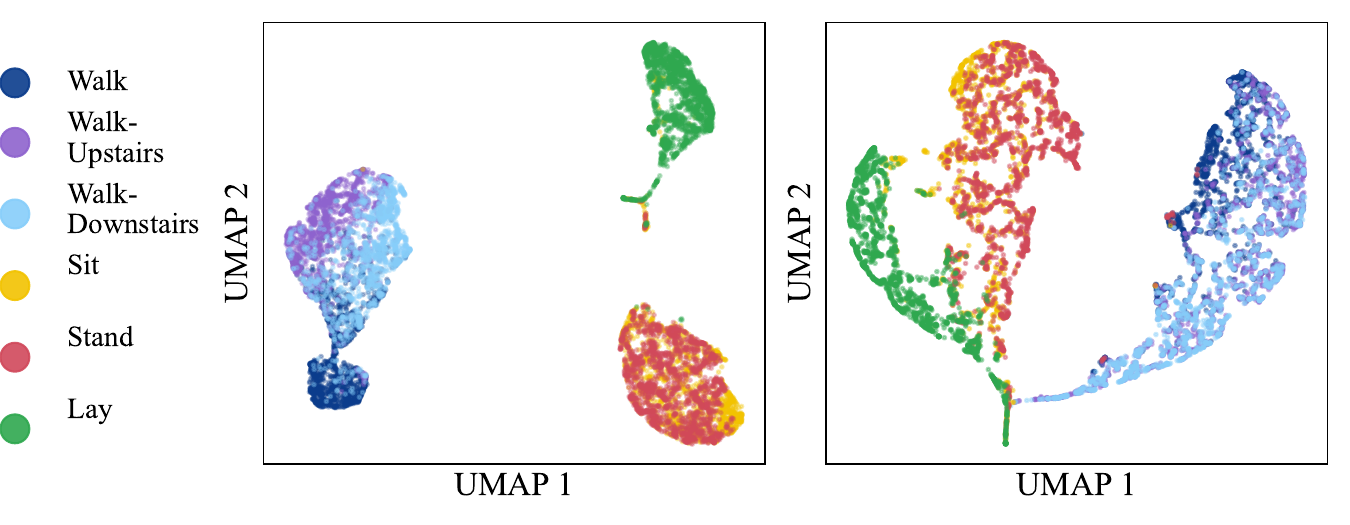}
    \caption{\textbf{Accelerometer Sensors.} Points are colored by the six activity labels, highlighting the separation of different activities in the low-dimensional space. \textbf{Right:} UMAP embedding using the two leading shared vectors. \textbf{Left:} UMAP embedding using the combined vectors: the two leading components from each sensor-specific differential vector and the two leading components of the shared vector.}

   \label{fig: Acc-gyr uniqe-vs-shared}
\end{figure}

Table~\ref{tab:ACC-body-gravity} summarizes the clustering performance across methods. Clustering based solely on DELVE’s differential vectors consistently achieves higher ARI and NMI scores than both Shnitzer and FKT, while concatenating shared and differential vectors gives performance comparable to Shnitzer. These results show that the differential components capture complementary information that is neither redundant nor negligible relative to the shared vectors, substantially improving clustering performance compared with using shared vectors alone.


\begin{table}[tb]
\centering
\begin{tabular}{lrrrrrr}
\toprule
 & \multicolumn{2}{r}{DELVE} & \multicolumn{2}{r}{Shnitzer +} & \multicolumn{2}{r}{FKT} \\
 \toprule
 & ARI & NMI & ARI & NMI & ARI & NMI \\
Vector &  &  &  &  &  &  \\
\midrule
$\psi^A$ /\ real & 0.291 & 0.413 & 0.297 & 0.382 & 0.275 & 0.432 \\
$\psi^B$ /\ imag & \textbf{0.413} & 0.475 & 0.226 & 0.387 & 0.000 & 0.004 \\
shared & 0.384 & \textbf{0.526} & \textbf{0.400} & \textbf{0.538} & 0.275 & 0.432 \\
all & \textbf{0.551} & \textbf{0.626} & \textbf{0.514} & \textbf{0.631} & 0.304 & 0.482 \\
\bottomrule
\end{tabular}

\caption{\textbf{Accelerometer sensors results}. The ARI and NMI of the K-means clustering, done on: (i)-(ii) two leading differential vectors of each method, (iii) two leading vectors of the shared component, and (iv) concatenation of all six vectors. Higher results are better.}\label{tab:ACC-body-gravity}
\end{table}

\section{Discussion}

In this work, we introduced DELVE, a spectral algorithm for extracting modality-specific latent structure from multimodal datasets. Our method addresses an important gap in the existing literature: while most multimodal algorithms focus on shared latent variables or building joint embeddings, fewer approaches explicitly target latent factors that are visible in one modality but absent in another. We demonstrated that differences in connectivity patterns across modalities contain structured information about these differential components, and that carefully designed graph filters can isolate this information.
Theoretically, we established convergence guarantees under a product manifold model, showing that DELVE recovers eigenfunctions associated with modality-specific coordinates.  
Extensive experiments on synthetic and real data illustrate that DELVE isolates directional modality-specific structure in contrast to previous spectral-based methods, including \cite{shnitzer2019recovering} and FKT.

Despite these strengths, several limitations remain. While our experiments demonstrate that DELVE is robust across a range of settings, systematic strategies for selecting parameters, particularly graph bandwidth and the spectral threshold used in the filter, could further enhance performance and interpretability. A deeper theoretical analysis of the iterative scheme would also be valuable.

Future work will extend DELVE in several directions. One promising direction is to couple the spectral approach with non-Euclidean or learned kernel metrics, which may tailor the connectivity structure to the underlying modality. Another potential avenue is integrating our method with supervision, allowing semisupervised or weak labels to guide the extraction of differential variables. Finally, applying DELVE to additional modalities, such as multimodal neuroimaging, remote sensing, or genomics, may offer further insights into modality-specific structure and reveal new use cases where existing shared-structure methods fall short.

Overall, our results show that differential latent variables are not merely noise to be discarded but contain meaningful, modality-specific information. DELVE provides a principled framework, supported by theory and empirical evidence, for isolating and analyzing this distinctive structure. We hope DELVE provides a useful tool for differential analysis in multimodal data and inspires further research into dataset-specific structure in manifold learning.

\subsection*{Reproducibility}
The code to reproduce all experiments in this paper is publicly available at \url{https://github.com/ShiraAlon/DELVE}.

\subsection*{Acknowledgments}

We thank Roy Lederman for assistance with datasets and Yoel Shkolnisky for his detailed and helpful feedback on a draft of this manuscript.
AM is supported in part by ISF Grant No. 1662/22 and NSF-BSF Grant No. 2022778.

\bibliographystyle{elsarticle-harv} 
\bibliography{bibliography}

\appendix

\section{Details of the compared methods}\label{app: methods}

\paragraph{Fukunaga-Koontz Transform (FKT)} The FKT identifies the directions that best separate two datasets by analyzing their covariance structures. In the graph setting, it computes the leading eigenvectors of
\begin{equation}
    \Tilde{L}_{FKT}^A = (L^A + L^B)^{-1}L^B, \qquad 
    \Tilde{L}_{FKT}^B = (L^A + L^B)^{-1}L^A,
\end{equation}
where \( L = D - W \) is the unnormalized Laplacian. Intuitively, it highlights patterns where the graphs differ most, revealing differential connectivity structures. 
We denote the $k$-th leading eigenvectors of $\Tilde{L}_{FKT}^A$ and $\Tilde{L}_{FKT}^B$ as $\delta(fk)^A_k$ and $\delta(fk)^B_k$, respectively. We will use these vectors to compare with the leading differential vectors, $\delta^A$ and $\delta^B$, obtained by DELVE.

\paragraph{Shnitzer et al. (2019)} This method recovers hidden components in multimodal data by constructing spectral operators that differentiate between shared and alternating latent variables, referred to as the $S$ and $A$ operators, respectively. The $A$ operator is antisymmetric, which means it possesses both real and imaginary eigenvalues. To recover the alternating hidden latent parameters, we consider the leading eigenvectors of $A$, focusing on both the real and imaginary parts. In this paper, we will denote these $k$-th leading vectors as $\delta(Sh)^r_k$ and $\delta(Sh)^i_k$.
Notably, Shnitzer et al. (2019) utilize a single alternating operator, $A$, to capture modality differences, without associating it with a specific sensor. In contrast, our approach constructs two separate operators, one for each modality, to extract modality-specific latent variables.

For all methods, we constructed the affinity matrices using Euclidean distances between the samples, with an adaptive bandwidth.


\section{Proofs}
\subsection{Proof of Theorem \ref{thm:convergence_symmetric_normalized}: eigenvector convergence rate for the symmetric normalized Laplacian}
\label{sec:convergence_proof_vksymm}
\begin{proof}
The eigenvectors of the symmetric-normalized Laplacian satisfy $v_k \propto D^{1/2}v_k^{(rw)}$, where $D$ is the diagonal matrix of degrees \cite[Prop 3.(2)]{von2007tutorial}.
We may freely choose a multiplicative constant for eigenvectors, so we take $v_k^{(rw)}$ $D$-normalized such that $(v_k^{(rw)})^T D v_k^{(rw)} = p n$ as in Theorem~5.5, \cite{cheng2022eigen}
 and  $v_k = \frac{D^{1/2}}{\sqrt{p n}} v^{(rw)}_k$ so that
 $\|v_k\|^2 = \frac{(v^{(rw)}_k)^T D v^{(rw)}_k}{p n} = 1$.
By the triangle inequality,
\begin{align}
    \Big\|v_k-  \alpha \phi_k(X)\Big\|
    &=
    \Bigg\|
        \frac{D^{1/2}}{\sqrt{p n}} v^{(rw)}_k
        -
        \alpha \phi_k(X)
    \Big\| \\
    &=
    \Bigg\|
        \frac{D^{1/2}}{\sqrt{p n}} v^{(rw)}_k
        -
        v_k^{(rw)}
        +
        v_k^{(rw)}
        -
        \alpha \phi_k(X)
    \Bigg\| \\
    &\leq
    \Bigg\|
        \frac{D^{1/2}}{\sqrt{p n}} v^{(rw)}_k
        -
        v_k^{(rw)}
    \Bigg\|
    + 
    \|
        v_k^{(rw)} - \alpha \phi_k(X)
    \|. \label{eq:abcd}
\end{align}
By \ref{thm:convergence_rw} \cite[Theorem 5.5]{cheng2022eigen}, the second term is bounded w.p. $> 1-4K^2n^{-10} - (4K+6)n^{-9}$ as
\begin{align} \label{eq:vkrw_bound}
    \|
        v_k^{(rw)} - \alpha \phi_k(X)
    \|
    =
    \OO
    \left(
        \epsilon_n + \epsilon_n^{-d/4-1/2} \sqrt{\log n / n}
    \right).
\end{align}
To bound the first term, we require an additional concentration bound on the degree matrix.
By \cite[Lemma 3.5]{cheng2022eigen}, for large values of $n$, with probability $>1-2n^{-9}$, the following holds uniformly for all $i$:
\begin{align}
    D_{ii}/ n
    =
    m_0 p
    +
    \OO
    \left(
        \epsilon_n
        +
        \epsilon_n^{-d/4}\sqrt{\log n/n}
    \right).
\end{align}
Where $m_0$ is a kernel-dependent constant and $p$ is the uniform density on the manifold.
For the Gaussian kernel $m_0=1$, so we can simply write
\begin{align} \label{eq:Dii_div_mu_n}
    D_{ii}/ p n
    =
    1
    +
    \OO
    \left(
        \epsilon_n
        +
        \epsilon_n^{-d/4}\sqrt{\log n/n}
    \right).
\end{align}
By taking the square root of both sides and substituting $\sqrt{1+x} = 1 + \OO(x)$ for $x \to 0$, we obtain that the same rate of convergence holds for the square root 
\begin{align}\label{eq:concentration_degree}
    \sqrt{D_{ii}/p n}
    =
    1
    +
    \OO \left(
        \epsilon_n
        +
        \epsilon_n^{-d/4}\sqrt{\log n/n}
    \right).
\end{align}
We now bound the first term of the right-hand side of Eq. \eqref{eq:abcd},
\begin{align} \label{eq:vknorm_times_Ddiff}
    \Bigg\|
        \frac{D^{1/2}}{\sqrt{p n}} v^{(rw)}_k
        -
        v_k^{(rw)}
    \Bigg\|
    =
    \Bigg\|
        \bigg(
            \frac{D^{1/2}}{\sqrt{p n}} - I
        \bigg)
        v_k^{(rw)}
    \Bigg\|
    \leq      
    \Bigg\|
        \bigg(
            \frac{D^{1/2}}{\sqrt{p n}} - I
        \bigg)
    \Bigg\|
    \big\|
        v_k^{(rw)}
    \big\|,
\end{align}
where the first term in the RHS denotes the operator norm.
The operator norm of a diagonal matrix is simply the maximum absolute value of the diagonal elements,
\begin{align} \label{eq:sqrtD_bound}
    \Bigg\|
        \bigg(
            \frac{D^{1/2}}{\sqrt{p n}} - I
        \bigg)
    \Bigg\|
    =
    \max_i
    \bigg|
        \frac{D_{ii}^{1/2}}{\sqrt{p n}}-1
    \bigg|
    =
    \OO
    \left(
        \epsilon_n
        +
        \epsilon_n^{-d/4}\sqrt{\log n/n}
    \right).
\end{align}
where this bound holds uniformly for all $i$ with probability at least $1-2n^{-9}$.
We now bound the norm of $v_k^{(rw)}$.
Recall that $v_k^{(rw)}$ is $D$-normalized such that
$(v_k^{(rw)})^T D v_k^{(rw)} / p n = 1$, hence
\begin{align}
     \frac{p n}{\max_i D_{ii}}
     \le
     \| v_k^{(rw)}\|^2
     \le
     \frac{p n}{\min_i D_{ii}}.
\end{align}
By \eqref{eq:Dii_div_mu_n} and by a first order expansion of $1/(1+x)$, both the upper and lower bounds are $1+\OO\left(\epsilon_n + \epsilon_n^{-d/4}\sqrt{\log n/n}\right)$, hence by a first order expansion of $\sqrt{1+x}$, we have
\begin{align}
    \| v_k^{(rw)}\| = \sqrt{\| v_k^{(rw)}\|^2} = 1+\OO\left(\epsilon_n + \epsilon_n^{-d/4}\sqrt{\log n/n}\right).
\end{align}
Plugging this and \eqref{eq:sqrtD_bound} back into \eqref{eq:vknorm_times_Ddiff} yields
\begin{align} \label{eq:vkrw_diff_bound}
       \Bigg\|
        \frac{D^{1/2}}{\sqrt{p n}} v^{(rw)}_k
        -
        v_k^{(rw)}
    \Bigg\|
    =
    \OO\left(\epsilon_n + \epsilon_n^{-d/4}\sqrt{\log n/n}\right).
\end{align}
Finally, substituting \eqref{eq:vkrw_bound} and \eqref{eq:vkrw_diff_bound} into \eqref{eq:abcd} and applying the union bound for the event that one of the bounds fails, we obtain that the following bound holds with probability at least $1-4K^2n^{-10} - (4K+8)n^{-9}$,
\begin{align}
   \|v_k - \alpha \phi_k(X)\| \label{eq:vksym_bound}
   =
   \OO\left(\epsilon_n + \epsilon_n^{-d/4-1/2}\sqrt{\log n/n}\right)  
\qquad \forall k \leq K.
\end{align}
\end{proof}

\subsection[Rate of convergence for alpha in theorems \ref{thm:convergence_rw} and \ref{thm:convergence_symmetric_normalized}]{Rate of convergence for $\alpha$ in theorems \ref{thm:convergence_rw} and \ref{thm:convergence_symmetric_normalized}}
\begin{lemma}\label{lem:alpha_convergence}
    Let $u,\phi$ be two vectors in $\R^n$ such that:
    \begin{enumerate}[label=(\roman*)]
        \item $\|u\|=1$
        \item $\|\phi\| = 1 + \OO( 1/\sqrt{n})$
        \item $\|u-\alpha \phi\| = \OO(\epsilon_n + \epsilon_n^{-d/4-1/2} \sqrt{\log n / n} )$ for $\epsilon_n = o(1)$
    \end{enumerate}
    Then it follows that 
    \begin{align}
        |\alpha| = 1 + \OO(\epsilon_n + \epsilon_n^{-d/4-1/2} \sqrt{\log n / n} ).
    \end{align}
\end{lemma}
\begin{proof}
By the inverse triangle inequality, and (iii), we have
\begin{equation}\label{eq:alpha_convergence_1}
    \big|
        \|u\| - |\alpha|\|\phi\|
    \big|
    \leq
    \|u-\alpha \phi\|
    =
    \OO(\epsilon_n + \epsilon_n^{-d/4-1/2} \sqrt{\log n / n} )
\end{equation}
From (i) and (ii), we have that
\begin{equation}\label{eq:alpha_convergence_2}
    \big|
        \|u\| - |\alpha|\|\phi\|
    \big|
    =
    \big|
        1
        -
        |\alpha|
        \left(
            1+\OO(1/\sqrt{n})
        \right)
    \big|
\end{equation}
By combining \eqref{eq:alpha_convergence_1} and \eqref{eq:alpha_convergence_2}, we obtain
\begin{align}
    |\alpha|
    \left(
        1+\OO(1/\sqrt{n})
    \right)
    =
    1
    +
    \OO(
        \epsilon_n
        +
        \epsilon_n^{-d/4-1/2} \sqrt{\log n / n}
    ).
\end{align}
Since $1/(1+x) = 1-x + \OO(x^2)$, we have that $1/(1+\OO(1/\sqrt{n})) = 1 + \OO(1/\sqrt{n})$. Dividing both sides by $1+\OO(1/\sqrt{n})$ we obtain
\begin{align}
    |\alpha|
    &=
    \left(
        1
        +
        \OO(
            \epsilon_n
            +
            \epsilon_n^{-d/4-1/2} \sqrt{\log n / n}
        )
    \right)
        \left(
        1+\OO(1/\sqrt{n})
    \right)
    \\
    &=
    1
    +
    \OO(
        \epsilon_n
        +
        \epsilon_n^{-d/4-1/2} \sqrt{\log n / n}
    ),
\end{align}
where the last equality follows from the fact that 
\begin{align}
    1/\sqrt{n} = \OO(\epsilon_n+\epsilon_n^{-d/4-1/2} \sqrt{\log n / n}).
\end{align}
\end{proof}

\subsection{Auxiliary lemmas for the proof of lemma \ref{lem:concentration} and Theorem \ref{thm:main_theorem}}

\begin{lemma}
\label{lem:lemma 1}
Let $v^A = u^A - \epsilon^A$ and $v^B = u^B -\epsilon^B$ be two vectors that satisfy:
\begin{enumerate}
    \item $\|v^A\|=\|v^B\|=1$.
    \item The vector $u^A$ is proportional to $u^B$ such that $u^A = cu^B$ for some constant $c$.
    \item The vectors norms $\|\epsilon^A\|,\|\epsilon^B\|$ are both $\OO(\epsilon_n)$.    
\end{enumerate}
Then $|(v^A)^T v^B| = 1-\OO(\epsilon_n)$.    
\end{lemma}

\begin{proof}

We begin by deriving a bound over $||u^A||$ and $||u^B||$ using the reverse triangle inequality, 
\begin{align}
  & ||v^A|| = ||u^A - \epsilon^A|| \geq ||u^A|| - ||\epsilon^A|| \notag \\& ||v^A|| \geq ||u^A || - ||\epsilon^A||. 
\end{align}
Thus, combined with assumptions (1) and (3), we get, \begin{align}\label{eq: u_A norm}
     ||u^A ||  \geq  1-  \OO(\epsilon_n).  
\end{align}
A similar bound can be derived for $||u^B||$.
Next, we show that $|(u^A)^T u^B| \geq 1 - \OO(\epsilon_n)$.
Since $u^A = c u^B$, then $\frac{|(u^A)^T u^B|}{||u^A|| ||u^B||} =  1$. Thus, combining \eqref{eq: u_A norm} and the above, we get
\begin{align}
    |(u^A)^T u^B| =||u^A|| \cdot ||u^B|| \geq 1 - \OO(\epsilon_n).
\end{align}
Finally, we derive a bound for $|(v^A)^T v^B|$. By the reverse triangle inequality
\begin{align}
    |(v^A)^T v^B|
    &=
    |(u^A - \epsilon^A)^T (u^B - \epsilon^B)| \\
    &\geq
    |(u^A)^T u^B| - |(\epsilon^A)^T u^B| - |(\epsilon^B)^T u^A| - |(\epsilon^A)^T \epsilon^B|.
\end{align}
We showed that the first term is lower bounded by $1-\OO(\epsilon_n)$. By the assumed bound on the norms of $\epsilon^A, \epsilon^B$, the fourth term is also bounded by $\OO(\epsilon_n^2)$. The second and third terms can be bounded by the Cauchy-Schwarz inequality as follows,
\begin{align}
    |(\epsilon^A)^T u^B| \leq ||\epsilon^A|| ||u^B|| = \OO(\epsilon_n).
\end{align}
Therefore, we have shown  that
\begin{align}
    |(v^A)^T v^B| \geq 1 - \OO(\epsilon_n).    
\end{align}
\end{proof}
\begin{lemma} \label{lem:aux_2}
    Let $V, U \in \R^{n \times K}$ be orthogonal matrices with columns $v_i$ and $u_i$ such that  
    \begin{align}
        v_i^T u_i \geq  1-\varepsilon  \qquad \forall i  = 1,\ldots, K, 
    \end{align}
    Then $\|V-U\|^2 \leq 2K\varepsilon$.    
\end{lemma}

\begin{proof}
If $v_i^T u_i \geq 1-\epsilon$ then $\|v_i-u_i\|^2 = 2(1-v_i^T u_i) \leq 2\varepsilon$.
The spectral norm of a matrix is bounded by its Frobenius norm. Thus,
\begin{align}
    \|V-U\|^2 \leq \|V-U\|_F^2 \leq 2K\varepsilon. 
\end{align}
\end{proof}

\begin{lemma} \label{lem:aux_3}
 Let $A \in \R^{n \times n}$ be a symmetric positive semi-definite matrix, with eigenvectors and eigenvalues $u_\ell,\lambda_\ell$. The spectral decomposition of $A$ is given by
 $A = \sum_\ell \lambda_\ell u_\ell u_\ell^T.$
 Let $V \in \R^{n \times K}$ be a matrix whose columns $\{v_i\}_{i=1}^K$ are orthogonal.
 If $|v_i^T u_j|\leq \epsilon \geq 1/\sqrt{n}$ for all $(i,j)$ then
 \begin{align}
    \|A - (I-VV^T)A(I-VV^T)\|\leq K\epsilon^2\sum_\ell \lambda_\ell.
 \end{align} 
\end{lemma}

\begin{proof}
Since $A$ is symmetric, the spectral norm of $A - (I-VV^T)A(I-VV^T)$ is equal to the largest eigenvalue, given by
\begin{equation}\label{eq:aux_lemma_max}
\max_{\|z\|=1} \Big|z^T \big(A - (I-VV^T)A(I-VV^T) \big) z\big|.  
\end{equation}
The proof of our claim is given in three steps:
\begin{enumerate}
    \item We prove that the maximizer $z^\ast$ of Eq. \eqref{eq:aux_lemma_max} is a linear combination of the vectors in $\{v_i\}$ such that $z^\ast = \sum_{j=1}^K \alpha_j v_j$. This implies that  
    \begin{align}\label{eq:aux_lemma_z_ast}
        (z^\ast)^T \big(A - (I-VV^T)A(I-VV^T) \big) z^\ast = (z^\ast)^T A z^\ast.
    \end{align}
    \item For every pair of vectors $v_i,v_j$ we derive the bound
    \begin{align}
        v_j^T A v_i \leq \epsilon^2 \sum_\ell \lambda_\ell.
    \end{align}
    \item Combining (i) + (ii), we show that the maximizer is bounded by 
    \begin{align}
        z^\ast A z^\ast = \sum_{i,j} \alpha_i \alpha_j \epsilon^2 \sum_l \lambda_l\,.
    \end{align}
     Together with the bound $\sum_{i,j} |\alpha_i| |\alpha_j|\leq K$ this concludes the proof.
\end{enumerate}
For step (i), any vector $z$ orthogonal to $V$ satisfies,
\begin{align}
    z^T \big(A - (I-VV^T)A(I-VV^T) \big) z = z^T A z - z^T A z = 0.
\end{align}
Thus, the matrix $\big(A - (I-VV^T)A(I-VV^T) \big)$ has a zero eigenvalue with multiplicity $n-K$. Since it is a symmetric matrix, it exhibits a spectral decomposition of $K$ eigenvectors with non-zero eigenvalues all in the span of $\{v_i\}_{i=1}^K$. This implies specifically that the leading eigenvector $z^\ast$ is a linear combination of $\{v_i\}_{i=1}^K$. Thus, $VV^T z^\ast = z^\ast$, and hence
\begin{align}
    (z^\ast)^T \big(A - (I-VV^T)A(I-VV^T) \big) z^\ast = (z^\ast)^T A z^\ast.
\end{align}
Step (ii) follows almost immediately from $|v_i^T u_j| \leq \epsilon$,
\begin{align}
    |v_j^T A v_i|
    =
    \Big|v_j^T \sum_{\ell=1}^n \lambda_\ell u_\ell u_\ell^T v_i\Big|
    \leq
    \sum_{\ell=1}^n \lambda_\ell |u_\ell^T v_j||u_\ell^T v_i|
    \leq
    \epsilon^2 \sum_{\ell=1}^n \lambda_\ell \,.
\end{align}
For step (iii), let $z^\ast = \sum_{i=1}^K \alpha_i v_i$. By Eq. \eqref{eq:aux_lemma_z_ast}, the maximal eigenvalue is bounded by
\begin{align}
    |z^\ast A z^\ast|
    =
    \Big|\sum_{ij} \alpha_i\alpha_j v_i^T A v_j \Big|
    \leq
    \sum_{ij} |\alpha_i\alpha_j v_i^T A v_j|  \leq
    \sum_\ell \lambda_\ell \epsilon^2 \sum_{ij} |\alpha_i \alpha_j|.
\end{align}
Finally, we have that $\sum_{ij} |\alpha_i \alpha_j| = \Big(\sum_i |\alpha_i|\Big)^2$. The maximum value of $\sum_i |\alpha_i|$, under the unit norm constraint $\sum_i \alpha_i^2=1$ is attained for $\alpha_i = 1/\sqrt{K}$, which gives $\sum_{ij} |\alpha_i \alpha_j| = K$.
\end{proof}


\begin{lemma}\label{lem:aux_4}
Let $M$ be a symmetric and positive semi-definite matrix. Let $V,U \in \R^{n \times K}$ be two orthogonal matrices such that $V^TV = U^TU = I$, and let $Q_1 = I-VV^T$ and $Q_2 = I- UU^T$ be two projection matrices. Then,
\begin{align}
    \|Q_1MQ_1-Q_2MQ_2\| \leq 4\|M\| \|U-V\|.
\end{align}
\end{lemma}

\begin{proof}
Since $M$ is symmetric, the spectral norm is equal to the largest eigenvalue, given by,
\begin{align}
    \max_{\|x\|=1} \Big| x^T(Q_1MQ_1-Q_2MQ_2)x\Big|.
\end{align}
As $M$ is positive semi-definite, it has a square root, denoted $M^{0.5}$.
Thus for any $x$ we have
\begin{align}\label{eq:aux_bound_qmq}
    |x^T(Q_1MQ_1-&Q_2MQ_2)x|
    =
    |x^TQ_1MQ_1x-x^TQ_2MQ_2x| 
    \\
    &=
    \big| \|M^{0.5}Q_1x\|^2-\|M^{0.5}Q_2x\|^2\big| \nonumber
    \\
    &=
    \big| \|M^{0.5}Q_1x\| + \|M^{0.5}Q_2x\| \big|\cdot\big|\|M^{0.5}Q_2x\|-\|M^{0.5}Q_1x\|\big|. \nonumber
\end{align}
By Cauchy-Shwartz we have $\|M^{0.5}Q_1x\| \leq \|M^{0.5}\|$ and $\|M^{0.5}Q_1x\|\leq \|M^{0.5}\|$. By the reverse triangle inequality,
\begin{align}\label{eq:aux_bound_mq}
\big|\|M^{0.5}Q_2x\|-\|M^{0.5}Q_1x\|\big| &\leq \big\|M^{0.5}(Q_1-Q_2)x\big\|
\\
&\leq \|M^{0.5}\| \|Q_1-Q_2\| = \|M^{0.5}\|\|VV^T-UU^T)\|. \nonumber
\end{align}
To bound the norm $\|VV^T-UU^T\|$, we apply the reverse triangle inequality again. For every vector $x$ we have,
\begin{align}\label{eq:aux_bound_vvt}
    \big |x^T(VV^T-UU^T)x \big |
    &\leq
    \big| \|V^Tx\|^2 - \|U^T x\|^2\big| \\
    &= \big(\|V^Tx\|+\|U^Tx\|\big)  \big(\|V^Tx\|-\|U^Tx\|\big). \nonumber \\
    &\leq 2 \big\|(V-U)^Tx\big\| \leq 2 \|V-U\|. \nonumber
\end{align}
Combining the bounds in Eqs. \eqref{eq:aux_bound_vvt}, \eqref{eq:aux_bound_mq} and \eqref{eq:aux_bound_qmq} yields
\begin{align}
    \|Q_1MQ_1-Q_2MQ_2\| \leq 4\|M\|\|V-U\|.
\end{align}
\end{proof}

\subsection{Proof of Lemma \ref{lem:concentration}}

\begin{proof}

By Theorem \ref{thm:convergence_symmetric_normalized}, with probability $>1-4K^2n^{-10} - (4K+8)n^{-9}$ there exists a constant $\alpha^A$ such that the unit-normalized eigenvector $v^A_{l,k}$ satisfies
\begin{equation} 
    \bigg\| \frac{\alpha^A}{\sqrt{np^A}} \rho_X(f^A_{l,k}) - v_{l,k}^A \bigg\|_2
    =
    \OO\left(\epsilon_n + \epsilon_n^{-d/4-1/2}\sqrt{\log n/n}\right), 
\end{equation}
for all eigenfunctions $f^A_{l,k}$ that correspond to eigenvalues $\eta^A_{l,k}$ that are among the smallest $K$ eigenvalues of the Laplace-Beltrami operator on $\M^A$.
The same result holds for $\alpha^B, f^B_{m,k'}, \eta^B_{m,k'}$ and $v^B_{m,k'}$.
We denote by $\epsilon_{l,k}^A$ the difference $\frac{\alpha^A}{\sqrt{n p_A}}\rho_X(f^A_{l,k}) -v_{l,k}^A$ and bound the inner product of $v_{l,k}^A$ and $v_{m,k'}^B$ using the triangle inequality, 
\begin{align}\label{eq:triangle}
|(v_{l,k}^A)^T v_{m,k'}^B| &=  
\bigg| \bigg(\frac{\alpha^A}{\sqrt{n p_A}} \rho_X(f^A_{l,k}) - \epsilon^A_{l,k} \bigg)^T  \bigg(\frac{\alpha^B}{\sqrt{n p_B}} \rho_X(f_{m,k'}^B) - \epsilon_{m,k'}^B \bigg) \bigg| \nonumber
\\
&\leq  \frac{\alpha^A\alpha^B}{n \sqrt{p_Ap_B}} |\rho_X(f_{l,k}^A)^T \rho_X(f_{m,k'}^B) \big | 
\\
&+
\frac{\alpha^B}{\sqrt{n p_B}}|(\epsilon_{l,k}^A)^T \rho_X(f^B_{m,k'})| 
+ 
\frac{\alpha^A}{\sqrt{n p_A}}|(\rho_X(f^A_{l,k}))^T \epsilon_{m,k'}^B| + 
|(\epsilon_{l,k}^A)^T \epsilon_{m,k'}^B|. \nonumber
\end{align}
The second and third terms in Eq. \eqref{eq:triangle} can be bounded via the Cauchy-Schwartz inequality. For example, the second term is bounded by
\begin{align}
    \frac{\alpha^B}{\sqrt{n p_B}}|(\epsilon_{l,k}^A)^T \rho_X(f^B_{m,k'})|
    \leq
    \frac{\alpha^B}{\sqrt{n p_B}}\|\epsilon_{l,k}^A\| \|\rho_X(f^B_{m,k'})\|.
\end{align}
Lemma 3.4 in \cite{cheng2022eigen}, w.p. $> 1 - 2K^2 n^{-10}$, the term $\frac{1}{np_B}\|\rho_X(f^B_{m,k'})\|^2$ is $1+\OO_p(\log(n)/n)$.
This result together with the bound on $\|\epsilon^A_{l,k}\|,\|\epsilon^B_{m,k'}\|$ in Theorem \ref{thm:convergence_symmetric_normalized}, and the concentration of $\alpha^B$ by Lemma \ref{lem:alpha_convergence} yields,
\begin{equation}\label{eq:convergence_rate_eps_rho}
\frac{\alpha^B}{\sqrt{n p_B}}|(\epsilon_{l,k}^A)^T \rho_X(f^B_{m,k'})| 
+ 
\frac{\alpha^A}{\sqrt{n p_A}}|(\epsilon_{m,k'}^B)^T \rho_X(f^A_{l,k})| = \OO\bigg( \epsilon_n + \sqrt{\frac{\log n}{n \epsilon_n^{d/2+1} }}\bigg).
\end{equation}
The fourth term is negligible,
\begin{equation}\label{eq:convergence_rate_eps_squared}
    |(\epsilon_{m,k'}^B)^T\epsilon_{l,k}^A|
    \leq
    \|(\epsilon_{m,k'}^B)\|\|\epsilon_{l,k}^A\|
    =
    \OO \Bigg( \bigg( \epsilon_n + \sqrt{\frac{\log n}{n \epsilon_n^{d/2+1} }}\bigg)^2 \Bigg)
\end{equation}
The first term of Eq. \eqref{eq:triangle} is equal to,
\begin{align}\label{eq:average}
    \tfrac{\alpha^A\alpha^B}{n\sqrt{p_Ap_B}}
    &\left|
        \rho_X(f^A_{l,k})^T \rho_X(f^B_{m,k'})
    \right| \\
    &=
    \tfrac{\alpha^A\alpha^B}{n\sqrt{p_Ap_B}}
    \left|
        \big(
            \rho_{\pi_1(X)}(f^{(1)}_l) \cdot \rho_{\pi_3(X)}(f^{(3)}_k)
        \big)^T
        \big(
            \rho_{\pi_2(X)}(f^{(2)}_m) \cdot \rho_{\pi_3(X)}(f^{(3)}_{k'})
        \big)
    \right| \nonumber
    \\
    &=
    \tfrac{\alpha^A\alpha^B}{n\sqrt{p_Ap_B}}
    \left|
        \sum_{i=1}^n
        f^{(1)}_l(\pi_1(x_i))
        f^{(2)}_m(\pi_2(x_i))
        f^{(3)}_k(\pi_3(x_i))
        f^{(3)}_{k'}(\pi_3(x_i))
    \right|. \nonumber
\end{align}
Consider the summands in Eq.~\eqref{eq:average}.
In our model, the coordinates on different manifolds are independently drawn. Thus, taking the expectation of the summands gives,\begin{align}\label{eq:expectation_break}
    &\E[f^{(1)}_l(\pi_1(x))f^{(2)}_m(\pi_2(x)) f^{(3)}_k(\pi_3(x)) f^{(3)}_{k'}(\pi_3(x))] \notag \\
    &= 
    \E[f^{(1)}_l(\pi_1(x))]
    \E[f^{(2)}_m(\pi_2(x))]
    \E[ f^{(3)}_k(\pi_3(x)) f^{(3)}_{k'}(\pi_3(x))]. 
\end{align}
In addition, by the orthogonality of eigenfunctions with different eigenvalues, we have,
\begin{align}\label{eq:zero_mean}
    &\E[f^{(1)}_l(\pi_1(x))] = 0  \quad \forall l > 0, \notag \\
    &\E[f^{(2)}_m(\pi_2(x))] = 0 \quad \forall m > 0, \\
    &\E[f^{(3)}_k(\pi_3(x)) f^{(3)}_{k'}(\pi_3(x))] = 0 \quad \forall (k \neq k'). \notag
\end{align}
As for the second moment, we have
\begin{align}\label{eq:unit_variance}
    &\E[(f^{(1)}_l(\pi_1(x)))^2] = 1  \quad \forall l, \notag \\
    &\E[(f^{(2)}_m(\pi_2(x)))^2] = 1 \quad \forall m, \\ 
    &\E\left[ \left(f^{(3)}_k(\pi_3(x)) f^{(3)}_{k'}(\pi_3(x))\right)^2 \right] \leq 1 \quad \forall (k, k').\notag
\end{align}
Eqs. \eqref{eq:expectation_break}, \eqref{eq:zero_mean} and \eqref{eq:unit_variance} imply that unless $l=m=0$ and $k=k'$, the sum in Eq.~\eqref{eq:average}  is over centred i.i.d random variables with variance $\le 1$.
Since the terms are i.i.d., the variance of the sum is the sum of variances and is thus bounded by $n$.
Therefore,
\begin{equation}\label{eq:convergence_rate_rho}
\frac{\alpha^A\alpha^B}{n\sqrt{p_Ap_B}}
\left|
    \rho_X(f_{l,k}^A)^T\rho_X(f_{m,k'}^B)
\right|
=
\frac{\alpha^A\alpha^B}{\sqrt{p_A p_B}}
\OO_p(1/\sqrt{n})
=
\OO_p(1/\sqrt{n}).
\end{equation}
It is easy to see that if $\epsilon_n < 1$ then $1/\sqrt{n} < \epsilon_n^{-d/4-1/2}\sqrt{\log n / n}$ and if $\epsilon_n \ge 1$ then $1/\sqrt{n} < \epsilon_n$.
In both cases, $1/\sqrt{n}$ is negligible with respect to the other terms (Eqs.  \eqref{eq:convergence_rate_eps_rho} and \eqref{eq:convergence_rate_eps_squared}).
Thus for $k \neq k'$, we have
\begin{align}
    (v^A_{l,k})^Tv^B_{m,k'}
    =
    \OO \left(
        \epsilon_n
        +
        \sqrt{\frac{\log n}{n \epsilon_n^{d/2+1} }}
    \right).
\end{align}
The only remaining case is when $l=m=0$ and $k=k'$.
Since $f_0^{(1)},f_0^{(2)}$ are both constant functions, then by Eq.~\eqref{eq:eigenfunctions_product}) we have
\begin{align}
    \rho_X(f_{0,k}^A) \sim \rho_{\pi_1(X)} (f_0^{(1)})\rho_{\pi_3(X)} (f_k^{(3)}) \sim \rho_{\pi_3(X)} f_k^{(3)}.
\end{align}
A similar derivation can be done for $\rho_X(f_{0,k'}^B)$ which implies that $\rho_X(f_{0,k}^A) \sim \rho_X(f_{0,k'}^B)$.
Let $u^A= \frac{\alpha^A}{\sqrt{n p_A}}\rho_X(f_{0,k}^A)$ and $u^B = \frac{\alpha^B}{\sqrt{n p_B}}\rho_X(f_{0,k}^B)$. We invoke Lemma \ref{lem:lemma 1} with $u^A,u^B$ and $\epsilon_{0,k}^A,\epsilon_{0,k}^B$ to obtain that
\begin{align}
    |(v_{0,k}^A)^T v_{0,k}^B| = 1-\OO(\|\epsilon^A_{0,k}\| +  \|\epsilon^B_{0,k}\|).
\end{align}
This result together with the bounds on $\|\epsilon^A_{0,k}\|,\|\epsilon^B_{0,k}\|$ from Theorem \ref{thm:convergence_symmetric_normalized} concludes the proof.

\end{proof}

\subsection{Proof of Theorem \ref{thm:main_theorem}} 
\label{sec:three_part_proof}
\begin{proof}
    
\textbf{Step 1: } Recall that $V^{B_\theta}$ contains the eigenvectors of $P^B_\tau$, associated with the shared variables. Thus, the matrix $E_1 = Q^{B_\theta} P^B_\tau Q^{B_\theta}$ contains the symmetric projection of $P^B_\tau$ onto the orthogonal complement of its leading eigenvectors $\{v^B_{0,1}, v^B_{0,2}, \ldots, v^B_{0,K}\}$, which eliminates those eigenvectors. 
Due to the eigenvalue structure of the product manifold $\M_B$ (see Eq.~\eqref{eq:eigenvalues_product}), it follows that for a large enough $K$, the leading eigenvector of $E_1$ is thus $v^B_{1,0}$, which is the leading eigenvector that is \textit{not} associated with $\theta$.
\begin{align}
    \argmax_{\|v\|=1} v^TE_1v = v^B_{1,0}.
\end{align}
By Theorem~\ref{thm:convergence_symmetric_normalized} we obtain that as $n \to \infty$,
\begin{equation}\label{first result}
    \bigg\|v_{1,0}^B -  \frac{\alpha}{\sqrt{p n}} \rho_X(f^B_{1,0})\bigg\| = \OO \bigg(\epsilon_n + \sqrt{\frac{\log n}{n \epsilon_n^{d/2+1}}}\bigg).
\end{equation}
\noindent \textbf{Step 2: } To bound the spectral norm $\|E_2 - E_1\|$, we add and subtract $Q^{A_\psi}Q^{B_\theta} P^B_\tau Q^{B_\theta}Q^{A_\psi}$ and then apply the triangle inequality,
\begin{align}\label{eq:term_1_term_2_sep}
\|Q^{A_\psi}Q^{A_\theta} P^B_\tau Q^{A_\theta}Q^{A_\psi}  - Q^{B_\theta} P^B_\tau Q^{B_\theta}\| 
&\leq 
\underbrace{\|Q^{A_\psi}Q^{A_\theta} P^B_\tau Q^{A_\theta}Q^{A_\psi} - Q^{A_\psi}Q^{B_\theta} P^B_\tau Q^{B_\theta}Q^{A_\psi} \|}_\text{\RNum{1}} 
\notag \\
&+
\underbrace{\|Q^{A_\psi}Q^{B_\theta} P^B_\tau Q^{B_\theta} Q^{A_\psi}  - Q^{B_\theta} P^B_\tau Q^{B_\theta}\|}_\text{\RNum{2}}
\end{align}
We begin by bounding \RNum{1},
\begin{align*}
\nonumber \|&Q^{A_\psi}Q^{A_\theta} P^B_\tau Q^{A_\theta}Q^{A_\psi}  - Q^{A_\psi}Q^{B_\theta}P^B_\tau Q^{B_\theta}Q^{A_\psi}\| 
=
\|Q^{A_\psi} ( Q^{A_\theta} P^B_\tau Q^{A_\theta} - Q^{B_\theta} P^B_\tau Q^{B_\theta}) Q^{A_\psi}\|
\\ 
&\leq 
\|Q^{A_\theta} P^B_\tau Q^{A_\theta} - Q^{B_\theta} P^B_\tau Q^{B_\theta}\| 
\leq
4\|P^B_\tau\| \|V^{A_\theta} - V^{B_\theta}\|
\end{align*}
The first inequality follows from the inequality $\|AB\| \le \|A\| \|B\|$ that holds for all operator norms, and 
since $Q^{A_\psi}$ is a projection matrix, we have $\|Q^{A_\psi}\| \le 1$.
The second inequality is proved in Lemma \ref{lem:aux_4}.
Combining lemma \ref{lem:concentration} with Lemma \ref{lem:aux_2} gives
\begin{align}
    \|V^{B_\theta} - V^{A_\theta}\|^2
    = 
    \OO_p\bigg(K\epsilon_n + K\sqrt{\frac{\log n}{n \epsilon_n^{d/2+1}}}\bigg).
\end{align}
This, combined with the bound on the eigenvalues of the operator, $\|P^B\| \leq 1$, yields the following inequality:
\begin{equation}\label{eq:term_1_bound}
\|Q^{A_\psi}Q^{A_\theta} P^B Q^{A_\theta}Q^{A_\psi}  - Q^{A_\psi}Q^{B_\theta} P^B Q^{B_\theta}Q^{A_\psi}\|^2 \leq \OO(K\epsilon_n) + \OO\bigg(K\sqrt{\frac{\log n}{n \epsilon_n^{d/2+1}}}\bigg).
\end{equation}
Next, we bound term \RNum{2} of Eq.~\eqref{eq:term_1_term_2_sep}.  We apply Lemma \ref{lem:aux_3} with the  matrix $Q^{B_\theta}P^B_\tau Q^{B_\theta}$ and projection $Q^{A_\psi}$. The bound $\epsilon$ on the inner products between a vector in the span of $Q^{B_\theta}$ and the eigenvectors of $P^B_\tau$, as required in the Lemma, is given by Lemma \ref{lem:concentration}. This yields
\begin{align}\label{eq:term_2_bound}
    \|
        Q^{A_\psi}Q^{B_\theta} P^B_\tau Q^{B_\theta}Q^{A_\psi}
        &-
        Q^{B_\theta} P^B_\tau Q^{B_\theta}
    \|
    \\
    &=
    \left(
        \sum_{l,k; \lambda_{l,k}^A <\tau}(1-\lambda_{l,k}^A)
    \right)
    \OO\bigg(K\epsilon_n
    +
    K\sqrt{\frac{\log n}{n \epsilon_n^{d/2+1}}}\bigg) \notag
    \\ 
    &=\OO(K^2\epsilon_n)
    +
    \bigg(K^2\sqrt{\frac{\log n}{n \epsilon_n^{d/2+1}}}\bigg). \nonumber
\end{align}
The convergence in Eq. \eqref{eq:term_1_bound} is slower than Eq. \eqref{eq:term_2_bound}. Thus, the overall convergence rate of $\|E_2 -E_1\|^2$ is,
\begin{align}
    \|E_2 -E_1\|^2 =  \OO(K\epsilon_n) +  \OO\bigg(K\sqrt{\frac{\log n}{n \epsilon_n^{d/2+1}}}\bigg)
    \label{e2-e1 result}.
\end{align}

\noindent \textbf{Step 3: } We apply the Davis-Kahan Theorem to the symmetric matrices $E_1$ and $E_2$ and their respective leading eigenvectors $\delta^B$ and $v^B_{1,0}$. According to the theorem, given that $(\delta^B)^T v^B_{1,0} \geq 0$, the following inequality holds \cite[Corollary 1]{yu2015useful}:
\begin{align}
    \| \delta^B - v^B_{1,0}\|^2
    \leq
    \frac{2^{\frac{4}{3}}\|E_2 - E_1\|^2}{\gamma_K^2}
\end{align}
where $\gamma_K$ is the minimal spectral gap, which under assumption (ii) is larger than zero.
Combining the outcome of step 2 (Eq. \eqref{e2-e1 result}) with the Davis-Kahan theorem yields, 
\begin{equation}\label{second result}
    \| \delta^B - v^B_{1,0}\|^2 \leq 
    \OO(K\epsilon_n) +  \OO\bigg(K\sqrt{\frac{\log n}{n \epsilon_n^{d/2+1}}}\bigg)
\end{equation}

Finally, the outcome of step 1 (Eq. \eqref{first result}) and step 3 (Eq.~\eqref{second result}) completes the proof.

\begin{align}
    \| \delta^B - \frac{\alpha}{p\sqrt{n}}\rho_{\pi_2(x)}(f^{(2)}_{1})\|^2
    &=
    \OO
    \left(
        \| \delta^B -  v^B_{1,0}\|^2) + \|v_{1,0}^B -  \tfrac{\alpha}{\sqrt{p n}}\rho_X(f_{1,0})\|^2
    \right) \nonumber\\
    &\leq \OO \bigg( K\epsilon_n  + K\sqrt{\frac{\log n}{n \epsilon_n^{d/2+1}}}\bigg).
\end{align}

\end{proof}

\section{Additional experiments}\label{app:additional_exp}

\subsection{ECG}

We reproduce the experimental framework of ~\citet{shnitzer2019recovering}, for multichannel ECG analysis. In line with their experimental setup, we examine both a synthetic fetal ECG scenario and actual ECG recordings.
In this experiment, we also compare our method to Independent Component Analysis (ICA), which provides a natural baseline for linear signal separation.

\subsubsection{Synthetic ECG Experiment}\label{sec:ECG synthetic}


\paragraph*{Data} 
We follow the preprocess and data generation scheme described in ~\citet[Subsection~6.2]{shnitzer2019recovering} using their data \footnote{\url{https://github.com/shnitzer/Recovering-hidden-components-in-multimodal-data/blob/master/main.m}}. The input consists of two abdominal ECG (ta-mECG) leads,  synthetically generated as a linear mixture of three real ECG recordings, denoted 
\( z_i^{(1)}, z_i^{(2)}, z_i^{(3)} \). 
The recording \( z_i^{(1)} \) is the maternal ECG component and \( z_i^{(2)}, z_i^{(3)} \) are the fetal ECG activities.
We generated $n=10^4$ pairs of mixed signals via,
\begin{equation}
    s_i^{A} = 2z_i^{(1)} - z_i^{(2)}, \qquad
    s_i^{B} = z_i^{(1)} - 0.5z_i^{(3)}.
\end{equation}
Thus, the maternal component is shared between the two signals, while the fetal components are distinct.
Figure~\ref{fig:synthetic_ecg} illustrates an example of the resulting simulated ta-mECG signals.
\begin{figure}[tb]   
   \includegraphics[width=0.95\linewidth]{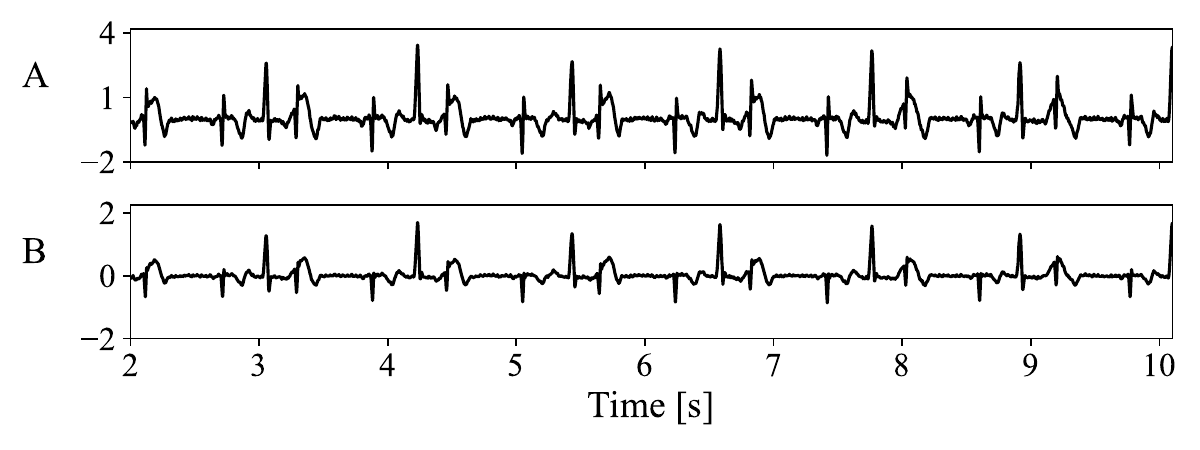}
   \caption{\label{fig:synthetic_ecg} Two simulated ECG leads representing two ta-mECG recordings.}
   
\end{figure}
The two input matrices $X^A$ and $X^B$ are 
computed by a lag map constructed from each signal, using windows of 8 samples with an overlap of 7 samples, resulting in a matrix of size $(T/7-1) \times 8$ for each signal, where $T$ is the number of time samples.  

\paragraph*{Evaluation Metric}

For evaluation, we analyzed the leading vectors computed by each method. 
Each vector was rescaled to the original signal length by assigning the value of its $i$‑th entry to the center of the corresponding $i$‑th window. Predicted vectors were then compared to fetal heartbeats in $z_i^{(3)}$. Since the predictions are continuous and the ground truth is binary, performance was evaluated using precision–recall area under the curve (PR-AUC). This metric summarizes the trade-off between precision and recall by computing precision and recall at all decision thresholds, and the area under the resulting curve provides a single scalar measure of detection performance. 
To account for small timing differences, a tolerance window was applied: a detection within this window of a true heartbeat was considered a true positive, a true heartbeat not detected within the window was counted as a false negative, and any detection outside a tolerance window of a true heartbeat was counted as a false positive. True negatives were counted normally.


\paragraph*{Results}

Figure~\ref{fig:synthetic_ecg_results} shows scatter plots of the outcome of all methods. Each observation is located according to its corresponding elements in the two leading differential vectors computed by each method. 
The plots in the upper row are colored according to one of the fetal ECG signals, and the plots in the second row are colored according to the maternal ECG signals. 
For DELVE, the outcome clearly separates points that correspond to fetal activity from the rest of the data.


Table~\ref{tab:results_synthetic_ecg_tol50_vecs} reports PR-AUC values across different tolerance windows. Both DELVE's vectors and the second vector of the ICA method consistently achieve the highest performance, even at small tolerance values. Shnitzer’s real component improves at higher tolerances, whereas its imaginary component and both FKT vectors show limited or delayed gains.

\begin{table}[ht]
\centering
\scalebox{0.9}{
\begin{tabular}{lrrrr}
\toprule
Tolerance & 0 & 10 & 25 & 50 \\
Method &  &  &  &  \\
\midrule
DELVE (A) & 0.031 & \textbf{0.985} & \textbf{0.985} & \textbf{0.985} \\
DELVE (B) & 0.001 & 0.223 & \textbf{0.956} & \textbf{0.956} \\
Shintzer + (imag) & 0.009 & 0.156 & 0.158 & 0.166 \\
Shintzer + (real) & 0.004 & 0.657 & \textbf{0.914} & \textbf{0.917} \\
FKT (A) & 0.013 & 0.158 & 0.158 & 0.159 \\
FKT (B) & 0.002 & 0.118 & 0.646 & 0.646 \\
ICA (1) & 0.002 & 0.062 & 0.168 & \textbf{0.985} \\
ICA (2) & 0.037 & \textbf{0.985} & \textbf{0.985} & \textbf{0.985} \\

\bottomrule
\end{tabular}

}
\vspace{10pt}  
\caption{\textbf{Synthetic ECG}. Coverage-PR-AUC estimated on different tolerance window sizes, for the leading vectors of each method (higher is better).}
\label{tab:results_synthetic_ecg_tol50_vecs}
\end{table}

\begin{figure}\hspace{-10pt}
   \includegraphics[width=1\linewidth]{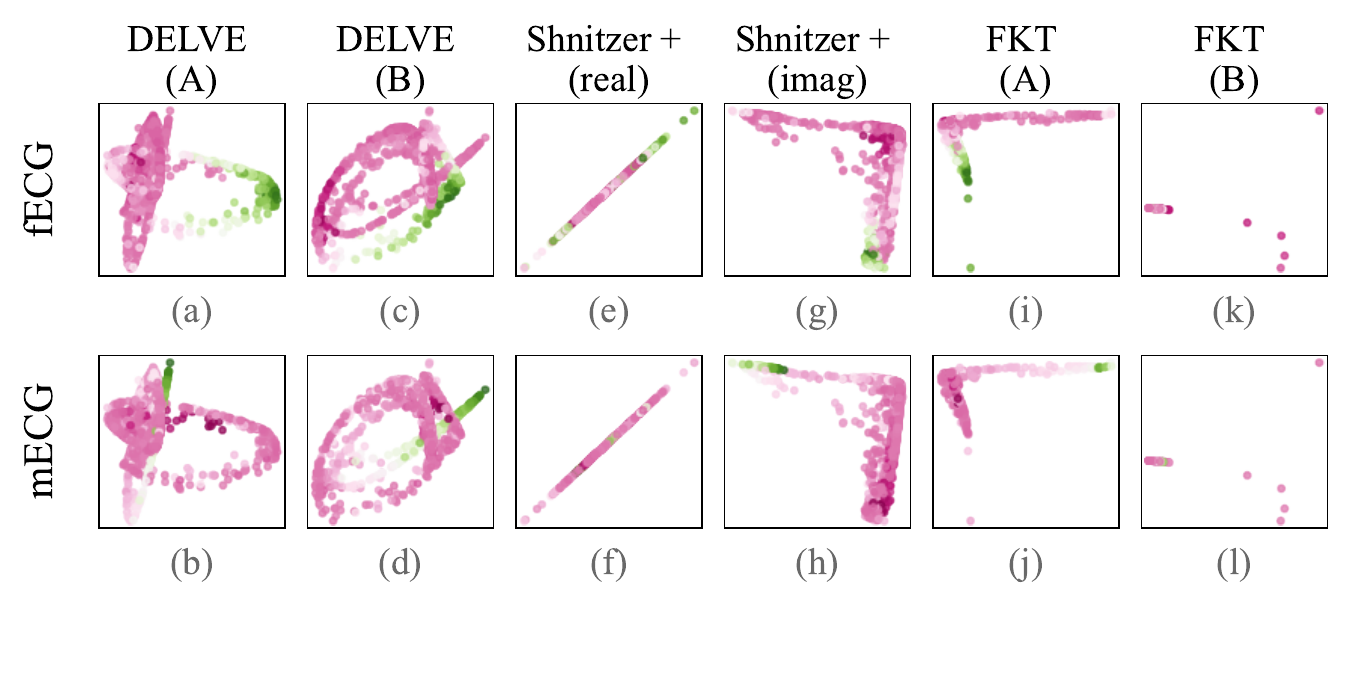}
   \caption{\textbf{Synthetic ECG experiment}. Scatter plots of the two leading vectors from each method and vector (e.g., panels (a), (b) show $\delta^A_0$ vs.\ $\delta^A_1$ and panels (c), (d) show $\delta^B_0$ vs.\ $\delta^B_1$).
    \textbf{Top:} Points colored according to the fetal ECG signal on a continuous scale, with green indicating an elevated signal associated with heartbeats.
    \textbf{Bottom:} Same vectors colored according to the maternal ECG signal.
   }
   \label{fig:synthetic_ecg_results}
\end{figure}

\subsubsection{Real ECG Experiment}

\paragraph*{Data}
We validated physiological data using the QT Database~\citep{laguna1997database,goldberger2000physiobank}. The QTDB contains annotated two-lead ECG signals sampled at 250\~Hz. 
Our preprocessing followed the steps outlined in~\citet[Subsection~6.3]{shnitzer2019recovering}, yielding recordings for 68 valid patients. 
Each record consists of two simultaneously measured ECG leads with expert-verified annotations. 
As in the synthetic example, the matrices $X^A$ and $X^B$ were computed using a lag map constructed from each signal, with windows of 8 samples and an overlap of 7 samples.


\paragraph*{Evaluation metric}


Similar to Section~\ref{sec:ECG synthetic}, each vector was rescaled, and then PR-AUC was computed against the annotated true fetal heartbeats using different tolerance windows.

\paragraph{Results}
Table~\ref{tab:ecg_summary} presents the mean, median, first and third quartiles, and the interquartile range (IQR) of the PR-AUC score using $50_{\text{ms}}$ and $100_{\text{ms}}$ tolerances.
The results for other tolerance values are summarized in boxplots, shown in Figure~\ref{fig: ECG real data one vec}.
DELVE consistently achieves the highest PR-AUC scores along with the second ICA vector. Shnitzer et al.'s operator and the FKT perform slightly lower. 


\begin{figure}\hspace{-5pt}
   \includegraphics[width=14.0cm]{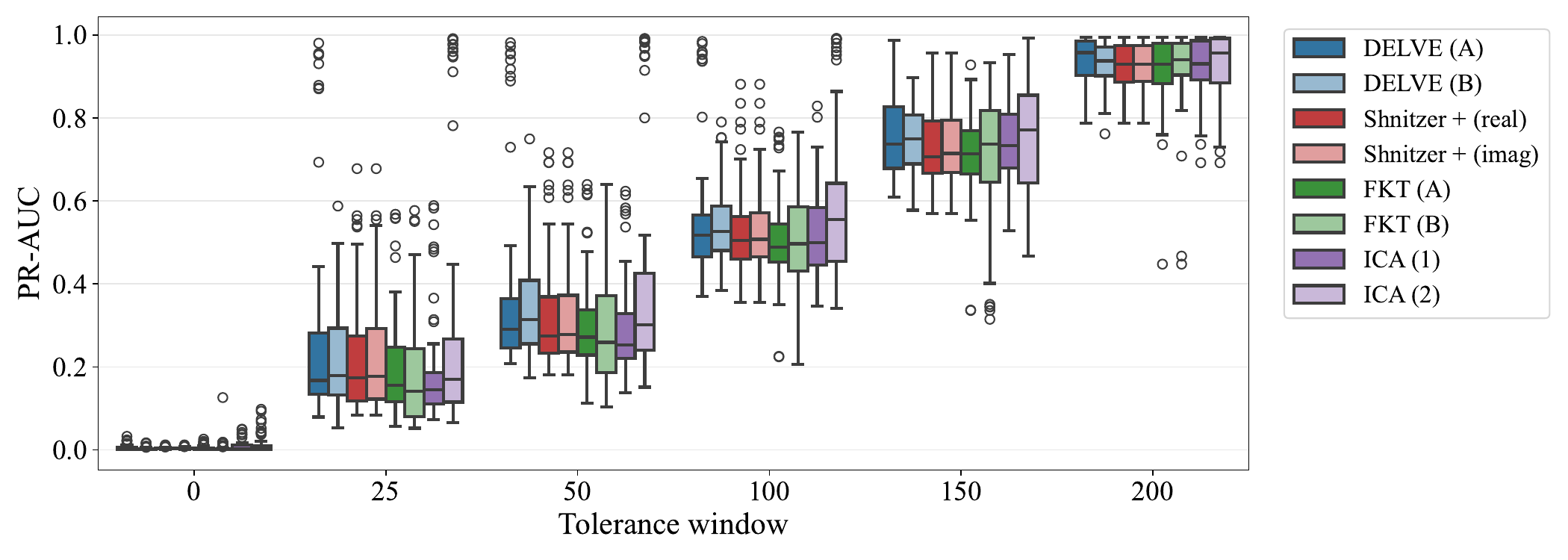}
   \caption{\textbf{ECG real data results}. Box plots showing the PR-AUC results using leading vectors of the different methods, tested on 68 valid subjects, for different tolerance values.}
   \label{fig: ECG real data one vec}
\end{figure}

\begin{table}[ht]
\centering
\scalebox{0.82}{
\begin{tabular}{lrllrrrrr}
\toprule
Tolerance & Method & Vector & Mean & Median & Q25 & Q75 & IQR \\
\midrule
50 & DELVE & A & \textbf{0.376} & \textbf{0.291} & 0.246 & 0.365 & \textbf{0.119} \\
50 & DELVE & B & \textbf{0.350} & \textbf{0.314} & 0.256 & 0.409 & 0.153 \\
50 & Shnitzer + &  imag & 0.331 & 0.278 & 0.236 & 0.373 & 0.137 \\
50 & Shnitzer + & real & 0.326 & 0.275 & 0.234 & 0.369 & 0.135 \\
50 & FKT &  A & 0.302 & 0.272 & 0.229 & 0.338 & \textbf{0.109} \\
50 & FKT &  B & 0.283 & 0.259 & 0.186 & 0.371 & 0.185 \\
50 & ICA &  1 & 0.294 & 0.253 & 0.221 & 0.329 & \textbf{0.108} \\
50 & ICA &  2 & \textbf{0.411} & \textbf{0.301} & 0.241 & 0.426 & 0.185 \\
\midrule
100 & DELVE &  A & \textbf{0.563} & \textbf{0.517} & 0.466 & 0.566 & \textbf{0.100} \\
100 & DELVE &  B & \textbf{0.546} & \textbf{0.526} & 0.481 & 0.588 & 0.108 \\
100 & Shnitzer + &  imag & 0.533 & 0.508 & 0.465 & 0.571 & \textbf{0.106} \\
100 & Shnitzer + & real & 0.527 & 0.505 & 0.460 & 0.562 & \textbf{0.102} \\
100 & FKT &  A & 0.508 & 0.488 & 0.452 & 0.545 & 0.093 \\
100 & FKT &  B & 0.498 & 0.497 & 0.431 & 0.587 & 0.156 \\
100 & ICA &  1 & 0.520 & 0.499 & 0.445 & 0.584 & 0.139 \\
100 & ICA &  2 & \textbf{0.597} & \textbf{0.555} & 0.455 & 0.643 & 0.188 \\
\bottomrule
\end{tabular}

}
\vspace{10pt}  
\caption{\textbf{Real ECG results}. Evaluation of coverage PR-AUC for each method on real QTDB data, based on 68 valid patients and two tolerance window settings. Performance is summarized using the mean, median, and interquartile range (IQR), with higher mean and median values indicating better performance and smaller IQR values indicating improved robustness. }
\label{tab:ecg_summary}

\end{table}

\end{document}